\documentclass[a4paper]{article}
\usepackage{amsmath}
\usepackage[english]{babel}
\usepackage{latexsym}
\usepackage{amssymb}
\usepackage{amscd}
\usepackage{amsgen,amstext,amsbsy,amsopn}

\newcommand{\version}{May 07, 2007}
\def\stackunder#1#2{\mathrel{\mathop{#2}\limits_{#1}}}
\def\dint{\mathop{\displaystyle \int}}%

\numberwithin{equation}{section}
\newcommand{\bdm}{\begin{displaymath}}
\newcommand{\edm}{\end{displaymath}}
\newcommand{\bdn}{\begin{eqnarray}}
\newcommand{\edn}{\end{eqnarray}}
\newcommand{\bay}{\begin{array}{c}}
\newcommand{\eay}{\end{array}}
\newcommand{\ben}{\begin{enumerate}}
\newcommand{\een}{\end{enumerate}}
\newcommand{\beq}{\begin{equation}}
\newcommand{\eeq}{\end{equation}}

\newcommand{\R}{\mathbb{R}}

\newcommand{\RT}{\mathbb{R}^3}

\newcommand{\xv}{\vec{x}}
\newcommand{\xvp}{\vec{x}^{\prime}}
\newcommand{\rv}{\vec{r}}
\newcommand{\diff}{\mathrm{d}}

\newcommand{\gpf}{\mathcal{E}^{\mathrm{GP}}_{g,\Omega}}
\newcommand{\tgpf}{\tilde{\mathcal{E}}_{\varepsilon,\omega}^{\mathrm{GP}}}
\newcommand{\togpf}{\hat{\mathcal{E}}_{\Omega,\omega}^{\mathrm{GP}}}
\newcommand{\gpd}{\mathcal{D}^{\mathrm{GP}}}
\newcommand{\gpe}{E^{\mathrm{GP}}_{g,\Omega}}
\newcommand{\tgpe}{\tilde{E}_{\varepsilon,\omega}^{\mathrm{GP}}}
\newcommand{\togpe}{\hat{E}_{\Omega,\omega}^{\mathrm{GP}}}
\newcommand{\gpm}{\phi_{g, \Omega}^{\mathrm{GP}}}

\newcommand{\tff}{\mathcal{E}^{\mathrm{TF}}_{g,\Omega}}

\newcommand{\ttffo}{\mathcal{E}^{\mathrm{TF}}_{1,\omega}}

\newcommand{\ttfe}{E^{\mathrm{TF}}_{1,\omega}}

\newcommand{\tfeinf}{E^{\mathrm{TF}}_{0,1}}

\newcommand{\ttfm}{\rho^\mathrm{TF}_{1,\omega}}

\newcommand{\ttfsupp}{\mathrm{supp}\left(\ttfm\right)}

\newcommand{\bi}{\mathcal{B}_{\varepsilon}^j}

\newcommand{\trial}{\tilde{\phi}}
\newcommand{\modu}{\tilde{\varrho}_{\varepsilon}}
\newcommand{\cut}{\chi_{\varepsilon}}

\newcommand{\phase}{g_{\varepsilon}}
\newcommand{\const}{C_{\omega}}

\newcommand{\latt}{\mathcal{L}}
\newcommand{\spac}{\ell_{\varepsilon}}

\newcommand{\magnp}{\vec{A}_{\omega/\varepsilon}}

\newcommand{\Wo}{W_{\omega}}

\newcommand{\xij}{|\vec{x}_i-\vec{x}_j|}
\newcommand{\eps}{\epsilon}

\newcommand{\half}{\mbox{$\frac{1}{2}$}}

\newtheorem{theorem}{Theorem}[section]
\newtheorem{lemma}[theorem]{Lemma}

\newtheorem{proposition}[theorem]{Proposition}

\newenvironment{proof}{\emph{Proof:}}{\begin{flushright} $ \Box $ \end{flushright}}

\begin{document}

\markboth{\scriptsize{BCPY \version}}{\scriptsize{BCPY \version}}

\title{The TF Limit for Rapidly  Rotating Bose Gases\\
 in Anharmonic Traps}
\author{J.-B. Bru${}^{a}$, M. Correggi${}^{b}$, P. Pickl${}^{a}$, J. Yngvason${}^{a, c}$\\ \normalsize\it \hspace{-.5 cm}\\ \hspace{-.5 cm}\normalsize\it  ${}^{a}$Fakult\"at f\"ur Physik, Universit{\"a}t Wien,      \normalsize\it
 Boltzmanngasse 5, 1090 Vienna, Austria\\ ${}^{b}$\normalsize\it  Scuola Normale Superiore SNS,
 \normalsize\it Piazza dei Cavalieri 7,
56126 Pisa, Italy
\\  ${}^{c}$\normalsize\it Erwin Schr{\"o}dinger
Institute for Mathematical Physics,\\ \normalsize\it Boltzmanngasse 9,
1090 Vienna, Austria}
\date{\version}

\maketitle

\begin{abstract}
Starting from the full many body Hamiltonian we derive the leading order energy and density asymptotics for the ground state of a dilute, rotating Bose gas in an anharmonic trap in the \lq Thomas Fermi' (TF) limit when the Gross-Pitaevskii coupling parameter and/or the rotation velocity tend to infinity. Although the many-body wave function is expected to have a complicated phase, the leading order contribution to the energy  can be computed by minimizing a simple functional of the density alone.
\end{abstract}

\section{Introduction}

Rotating  Bose-Einstein condensates exhibit fascinating quantum phenomena like superfluidity and
quantization of vorticity  and their study is currently an active area of both
experimental and theoretical research. Much of the theoretical work
(see, e.g., the monograph \cite{Afta} where an extensive list of references can be found)
is based on an  effective description of the ground state in terms of the  Gross-Pitaeveskii (GP)
equation that has recently been  proved  \cite{LS06, S03} to be exact in a suitable limit.
The corresponding result  for the non-rotating case was obtained in \cite{LiebSeiringerYngvason1},
but the rotating case was a long-standing problem whose solution required different techniques
from those of \cite{LiebSeiringerYngvason1}. In fact, the rotating case differs markedly
from the non-rotating one  since the absolute many-body  ground state is in general not
the same as the bosonic ground state \cite{S03}.

The limit considered in  \cite{LS06, S03} is the GP  limit of the many-body
ground state which means that the particle number $N$ tends to infinity while
the GP parameter $Na$, with $a$ the scattering length of the interaction potential,
as well as the rotational velocity, $\Omega$, are kept fixed.  In several  experiments
the GP parameter can be quite large, however,  and also the rotational velocity can be
so large that the effect of the rotation becomes comparable with that of the interactions.
Some of the interesting phenomena expected under such conditions are discussed in \cite{B, baym, ECHSC, Fe, FiB, FZ,KTU,KB, Lund1, WP}.
To describe such cases theoretically  it is natural to consider the limit when $Na$ and/or $\Omega$ tend to infinity.
Since the error estimates in \cite{LS06, S03} are not uniform in these parameters the results
of  \cite{LS06, S03}  do not apply to this situation and a complementary investigation is called for.
In order to allow arbitrarily large rotational velocities we shall
require that the  trap potential increases more rapidly at
infinity than quadratically (anharmonic traps). 

In a recent paper \cite{CRY06} the limit of GP theory  (but not the many-body problem) for large $Na$ and large $\Omega$
was studied for a rotating 2D gas in a \lq flat' trap with walls that confine the gas to a fixed bounded region. (See also \cite{CRY07} for an extension to more general 2D traps.) It was proved that the leading order energy and density
asymptotics is correctly described by an energy functional of the
density without a gradient term. This functional for rapidly
rotating gases was first introduced and studied in \cite{FiB}.
Because of its formal analogies with the Thomas-Fermi (TF)  theory
for  fermions it is also referred to as a TF functional although
the physical situation is quite different.

In the present paper we derive the TF description from the many-body problem in 3D.
This might at first sight appear to be a simple combination of the results of  \cite{LS06, S03}
with those of \cite{CRY06} (extended to 3D and more general traps): first take the GP limit of the many-body theory and
then the TF limit of the GP theory. But such an argument is not valid because the error estimates
in \cite{LS06} blow up if $Na$ and/or $\Omega$ tend to infinity. In fact, it is clear that
it will at least be  necessary to require explicitly  that the gas is {\it dilute} in the sense
that the average particle distance is much larger than the scattering length. (This condition
is automatically fulfilled if $Na$ is kept fixed.) It is also instructive to compare with the
proof of  dimensional reduction of Bose gases in tightly confining traps \cite{LSY, SY} where
one also has the option of taking a limit in two steps. Here the 1D or 2D limits of the 3D GP
theory do not, in fact,  cover all cases that can occur when one starts with the 3D many-body theory and takes  its  tight confinement limit directly. In  the present situation, however,  the TF limit of the GP theory exhausts the leading order asymptotics  for the many-body ground state,  provided only the gas remains dilute as the limit is taken. Because the many-body wave function will  have a complicated phase
it  is not self-evident  that the ground state energy can be computed accurately by a functional of the density alone. That this is indeed possible is a basic message of the present  paper.

 The main
techniques applied in our derivation of the TF limit are on the
one hand an extension  of \cite{CRY06} to 3D and
anharmonic, homogeneous  potentials, and on the other hand the
techniques described in \cite{LSSY05}  for treating the many-body problem in the non-rotating case.
Additional tools are the diamagnetic inequality as well as the
method used in \cite{S03} to bound the many-body energy from above
by the GP energy in the rotating case.  Our results concern the
ground state energy and density to leading order in small
parameters (the reciprocal  of the coupling constant or
the rotational velocity as well as the ratio of $a$ to the mean
particle spacing).  In contrast to \cite{LS06} Bose Einstein
condensation (BEC) is not proved. In fact, even in the
non-rotating  case a proof of BEC in the TF
limit is still an open problem.

The variational wave functions we employ for deriving
upper bounds to the energy reflect the expected shorts scale structure due to interactions as well as the quantized vortices generated by the rotation. In 3D the exact vortex structure is presumably rather complicated, e.g.,  due to bending of vortex lines \cite{Afta1}.  Our trial functions do not model such details but  they are sufficient for  leading order calculations and leave room for  improvements.

The paper is organized as follows. In the next section we describe
the general setting and state our main results about the energy
and density asymptotics. Section 3  is concerned with the TF limit
of GP theory, generalizing the results of \cite{CRY06} to 3D and
homogeneous, anharmonic trapping potentials, as well as some basic properties of the TF theory needed for the proofs of the quantum mechanical (QM) limit theorems. In Section 4 we prove
the upper and lower bounds on the many-body quantum-mechanical energy and obtain limit theorems for the density as corollaries. Some additional properties of the TF density are discussed in an Appendix.
\section{General Setting and Main Results\label{Section Main result}}

We consider a system of  $N$ spinless bosons in $%
\RT$  with mass $m$, trapped in  an external potential $V$ and interacting with a nonnegative, radially symmetric  pair
potential $v$ of finite range. The basic quantum mechanical Hamiltonian in a reference
frame that rotates with angular velocity $\vec \Omega $ is
\begin{equation}\label{ham1}
    H_{N}\equiv \stackrel{N}{\stackunder{j=1}{\sum }}\left( -\Delta _{j}- \vec{L}_{j}\cdot \vec{\Omega} +V(\xv_{j})\right) +\stackunder{1\leq i<j\leq N}{\sum } v(|\xv_{i}-\xv_{j}|).
\end{equation}
where units have been chosen so that  $2m=\hbar =1$. Also\footnote{We use the
following notation: $ \xv \equiv  (x,y,z)$ will always denote a
point in $ \RT $, while $ \vec{r} $ and $ z $ are its projections
on the $x,y-$plane and the $ z-$axis respectively (cylindrical
coordinates), namely $ \xv \equiv  (\rv,z) $. Moreover $ |\xv|
$ and $ r \equiv  |\rv| $ will denote the moduli of the
corresponding vectors, whereas $\zeta \equiv x+iy$ will be the
complex number associated with the $x,y-$coordinates of $ \xv$.} $
\xv_{j}\in \RT\ $and $\vec{L}_{j}\equiv - i \xv_{j} \wedge \nabla
_{j}$ for $j=1,...,N$ are respectively the positions and the
angular momentum operators of the particles. The Hamiltonian operates
on symmetric (bosonic) wave functions in $L^{2}\left(
\R^{3N},\diff\xv_{1}...\diff\xv_{N}\right) \ $.

We keep the external potential fixed and choose the associated length scale
(the extension of the ground state of $-\Delta+V$) as a length unit.
In order to be able to vary the scattering length\footnote{Alternatively,
one could keep $a$  fixed and vary the external potential and hence the length scale.
What matters is only the ratio of $a$ to this scale.} of the interaction
potential $v$ with $N$ we write $v(|\vec x|)=a^{-2}v_1(|\vec x|/a)$ where $v_1$
is a potential with scattering length 1.  Then $v$ has scattering length $a$ and
for fixed $v_1$ the Hamiltonian is parametrized by  $N$ and $a$ besides $\vec \Omega$.
For convenience, we include a factor $4\pi$ in the definition of the GP parameter
\beq\label{gpparam} g\equiv 4\pi Na.\eeq
The ground state energy  of \eqref{ham1}, i.e., the infimum of its spectrum,
will be denoted by $E_{g,\Omega }^{\mathrm{QM}}\left( N\right) $.

Introducing  the vector potential $\vec{A}_\Omega\equiv
\hbox{$\frac 12$}\Omega (\vec{e}_{z} \wedge \xv)$ associated with
a rotation $\vec\Omega=\Omega \vec{e}_{z}$, where  $\vec{e}_{z}$
 is the unit vector in
$z$-direction, the Hamiltonian $H_{N}$ can be rewritten in the
form
\begin{equation}\label{ham2}
    H_{N}=\stackrel{N}{\stackunder{j=1}{\sum }}\left( \left[ -i\nabla_{j}-\vec{A}_\Omega(\xv_{j}) \right]^{2}+V(\xv_{j})-\hbox{$\frac{1}{4}$}\Omega ^{2}r_{j}^{2}\right) +%
\stackunder{1\leq i<j\leq N}{\sum }v(|\xv_{i}-\xv_{j}|),
\end{equation}
where $r \equiv  \left| \vec{e}_{z}\wedge \xv\right| $ is the
distance from the axis of rotation. The splitting of the term $- \vec{L}\cdot \vec{\Omega}$
into the contribution of the vector potential and the term $-\Omega
^{2}r^{2}/4$ corresponds respectively to the Coriolis and the centrifugal force
in the rotating frame. The vector potential is primarily responsible for the formation
of vortices while the centrifugal potential affects the overall density profile if the rotational velocity is high enough.

In order  that \eqref{ham2} is bounded from below and trapping for
all  $\Omega $ we require the external potential $V$ to be bounded
from below and moreover that \beq\label{stable}V( \xv)
-\hbox{$\hbox{$\frac{1}{4}$}$}\Omega^{2}r^{2}\to \infty\quad \hbox{
for}\quad  |\vec x|\to\infty.\eeq
Due to the particle repulsion and the centrifugal force
the gas cloud expands as $g$ or $\Omega$ tend to infinity so essentially only the
behavior of $V$  for large $|\vec x|$ matters. For simplicity we
shall assume that $V$ is  a homogeneous
function\footnote{Asymptotic homogeneity in the sense of \cite{Lieb2}
would also suffice.} of order $s>2,$ i.e. $V\left( \lambda
\xv\right) =\lambda ^{s}V\left( \lambda \xv\right) $ for any
$\lambda >0$ and $\xv\in \RT$, but we need {\it not} assume that $V$ is symmetric w.r.t.\ rotations about the $z$-axis.  In order to obtain explicit error estimates in Section 3 we shall assume that $V$ is twice continuously differentiable but H\"older continuity would in fact be sufficient for the main results. 

The case of  a \lq flat' trapping potential, i.e., $V=0$ within a bounded, open set $\mathcal B$ with a smooth boundary and $\infty$ outside,  can also be treated by our methods. This case corresponds formally to $s=\infty$ and the formulas for the energy and density asymptotics can be obtained as the $s\to\infty$ limit of the formulas for finite $s$. At the end of Section 4 we shall comment on the results for flat traps and the minor modifications required of the proofs.

The \emph{GP functional} in the rotating frame is defined as
\begin{equation}
    \mathcal{E}^{\mathrm{GP}}_{g,\Omega}\left[ \phi \right] \equiv \int_{\RT}\diff%
\xv\left\{ \left| \left[ \nabla - i \vec{A}_\Omega\right] \phi
\right| ^{2}+V\left| \phi \right| ^{2}-\hbox{$\frac{1}{4}$}\Omega
^{2}r^{2}\left| \phi \right| ^{2}+g\left| \phi \right|
^{4}\right\}   \label{GP functional}
\end{equation}
on the domain
\begin{equation}
    \mathcal{D}^{\mathrm{GP}}\equiv \left\{ \phi \in L^{4}\left( \RT\right)
:V\left| \phi \right| ^{2}\in L^{1}\left( \RT\right) \mathrm{\ and\ }%
\left[ \nabla -i\vec{A}_\Omega\right] \phi \in L^{2}\left(
\RT\right) \right\} .
\end{equation}
The corresponding energy is
\begin{equation}
E_{g,\Omega }^{\mathrm{GP}}\equiv \stackunder{\phi \in \mathcal{D}^{\mathrm{GP}%
},\left\| \phi \right\| _{2}=1}{\inf }\mathcal{E}_{g,\Omega}^{\mathrm{GP}}\left[ \phi
\right].   \label{GP energy}
\end{equation}
The infimum is, in fact, a minimum and we denote  any normalized  minimizer by $\phi^{\rm GP}_{g, \Omega}$.  The corresponding density is $\rho^{\rm GP}_{g, \Omega}\equiv |\phi^{\rm GP}_{g, \Omega}|^2$. The minimizer may not be unique because vortices can break rotational symmetry, but  any minimizer  satisfies the variational (GP) equation
\beq
\left[ -(\nabla - i \vec{A}_\Omega)^2+V-\hbox{$\frac{1}{4}$}\Omega
^{2}r^{2}+2g\left| \phi^{\rm GP}_{g, \Omega} \right|
^{2}\right]\phi^{\rm GP}_{g, \Omega}=\mu_{g,\Omega}^{\rm GP}\phi^{\rm GP}_{g, \Omega} \label{GPeq}
\eeq
where $\mu_{g,\Omega}^{\rm GP}$ is the GP chemical potential. Multiplying \eqref{GPeq} with $\overline{ \phi^{\rm GP}_{g, \Omega}}$ and integrating gives
\beq\label{GPchempot}
\mu_{g,\Omega}^{\rm GP}=E_{g,\Omega }^{\mathrm{GP}}+2g\Vert \rho^{\rm GP}_{g, \Omega}\Vert_2.
\eeq

In  \cite{LS06, S03}  it is proved that $E_{g,\Omega }^{\mathrm{QM}}(N)/NE_{g,\Omega }^{\mathrm{GP}}\to 1$ as $N\to\infty$ if $g$ and $\Omega$ are fixed.
In the present paper we are concerned with the situation where $g$ and/or  $\Omega$ tend to infinity together with $N$. As we shall show, the first term in \eqref{GP functional} is negligible in this limit and the ground state energy and density can be described  in terms of  the \emph {TF functional}
\begin{equation}
    \tff\left[ \rho \right] \equiv \int_{\RT}\diff%
\xv\left\{ V\rho -\hbox{$\frac{1}{4}$}\Omega ^{2}r^{2}\rho +g\rho
^{2}\right\} \label{TF functional}
\end{equation}
defined on the domain
\begin{equation}
    \mathcal{D}^{\mathrm{TF}}\equiv \left\{ \rho \in L^{2}\left( \RT\right)
:\rho \geq 0,V\rho \in L^{1}\left( \RT\right) \right\}
\end{equation}
with the energy
\begin{equation}
    E_{g,\Omega }^{\mathrm{TF}}\equiv \stackunder{\rho \in \mathcal{D}^{\mathrm{TF}%
},\left\| \rho \right\| _{1}=1}{\inf }\tff \left[ \rho
\right] .  \label{TF energy}\eeq

The minimization problem \eqref{TF energy} has a unique solution given by
\begin{equation}
\rho _{g,\Omega }^{\mathrm{TF}}\left( \vec{x}\right)
=\dfrac{1}{2g}\left[
\mu _{g,\Omega }^{\mathrm{TF}}+\hbox{$\frac{1}{4}$}\Omega ^{2}r^{2}-V\left( \vec{x}%
\right) \right] _{+}, \label{TF density}
\end{equation}
where $[\cdot]_+$ denotes the positive part and $\mu _{g,\Omega }^{\mathrm{TF}}$ is the TF chemical
potential determined by the normalization $||\rho _{g,\Omega }^{\mathrm{TF}}||_{1}=1$. Multiplying \eqref{TF density} by $\rho _{g,\Omega }^{\mathrm{TF}}$ and integrating gives
\beq
\mu _{g,\Omega }^{\mathrm{TF}%
}= E_{g,\Omega }^{\mathrm{TF}}+2g||\rho _{g,\Omega }^{\mathrm{TF}%
}||_{2}^{2}. \label{TFchempot}
\eeq
By simple rescaling (explained at the beginning of Section 3) we obtain the relations
\begin{equation}
g^{-{s}/{(s+3)}%
}E_{g,\Omega }^{\mathrm{TF}} =E_{1,\omega }^{\mathrm{TF}}\quad\hbox{and}\quad
g^{-{s}/{(s+3)}%
}\mu_{g,\Omega }^{\mathrm{TF}} =\mu_{1,\omega }^{\mathrm{TF}}
\label{TF big interaction}
\end{equation}
with \beq 
\omega\equiv  g^{-{(s-2)}/{\left( 2s+6\right) }}\Omega\label{omega}
\eeq and likewise \beq\label{TFdensityscaling} g^{{3}/{(s+3)}}\rho _{g,\Omega
}^{\mathrm{TF}}\left( g^{{1}/{(s+3)}}\vec{x}\right)=\rho
_{1,\omega }^{\mathrm{TF}}\left(
\vec{x}\right). \eeq The case $\omega=0$
corresponds  to the standard TF functional without rotation
whose relation to the many-body problem was already discussed in
\cite{LSSY05}.  Note also that $E_{1,\omega }^{\mathrm{TF}}$ is a
decreasing function of $\omega\geq 0$ with range $(-\infty, E_{1,0
}^{\mathrm{TF}}]$ with $E_{1,0
}^{\mathrm{TF}}>0$. In particular, there is an $\omega>0$ such that
$E_{1,\omega }^{\mathrm{TF}}=0$.

In the case that $\omega$ tends to infinity the rotational term completely dominates the interaction term. In this case it is appropriate to scale lengths with $\Omega^{2/(s-2)}$ rather than $g^{1/(s+3)}$ (cf. Section 3), obtaining
\begin{equation}
\Omega^{-2s/(s-2)}E_{g,\Omega }^{\mathrm{TF}} =E_{\gamma,1}^{\mathrm{TF}}\quad\hbox{and}\quad
\Omega^{-2s/(s-2)}\mu_{g,\Omega }^{\mathrm{TF}} =\mu_{\gamma,1}^{\mathrm{TF}}
\label{TFenergyscalingstrong}
\end{equation}
with \beq\label{gamma} \gamma\equiv  \Omega^{-2(s+3)/(s-2)}g=\omega^{-2(s+3)/(s-2)}
\eeq and  \beq \Omega^{6/{(s-2)}}\rho _{g,\Omega
}^{\mathrm{TF}}\left(\Omega^{
2/{(s-2)}}\vec{x}\right)=\rho _{\gamma,1
}^{\mathrm{TF}}\left(\vec{x}\right).\label{TFdensityscalingstrong} \eeq
Moreover, as $\omega\to \infty$, i.e., $\gamma\to 0$, we have
\begin{equation}
\lim_{\gamma\to 0 }E_{\gamma,1}^{\mathrm{TF}} =E_{0,1}^{\mathrm{TF}}=%
\stackunder{\xv \in \mathbb{R}^{3}}{\inf }\left\{
V(\xv)-\hbox{$\frac{1}{4}$}r^{2}\right\}<0  \label{TF strong rotation}
\end{equation}
while $\rho _{\gamma,1
}^{\mathrm{TF}}$ converges
 to a measure supported on the set $\mathcal{M}$ of
minima of the function $W(\vec x)\equiv V(\xv)-\hbox{$\frac{1}{4}$}r^{2}$. This together with
the other facts mentioned about the TF theory is discussed further
in Section 3.2.

The scaling properties of the TF theory already suggest that one should distinguish between the following three cases when the $N\to\infty$
limit of the many-body ground state with $g$ and/or $\Omega$ also tending to infinity is considered:
\begin{itemize}
\item {\bf Slow or moderate rotation}, $\omega\ll 1$: The effect of the rotation is negligible to leading order.\footnote{This regime could be subdivided further into 'slow' rotations, where vortices do not yet form, and 'moderate' rotations where vortices are present. This finer division goes, however, beyond the leading order considerations of the present paper.}
\item {\bf Rapid rotation}, $\omega\sim 1$: Rotational effects are comparable to those of the interactions.
\item {\bf Ultrarapid rotation}, $\omega\gg 1$: Rotational effects dominate.
\end{itemize}

 Moreover, a description in terms of the simple density functional \eqref{TF functional} can only be expected in a dilute limit.
 A convenient  measure for diluteness turns out to be smallness of the parameter
 \beq a^3N\Vert \rho_{g,\Omega }^{\mathrm{TF}}\Vert_{\infty }\sim N^{-2}g^3\Vert \rho_{g,\Omega }^{\mathrm{TF}}\Vert_{\infty}.\eeq

Our main result about the energy asymptotics is the following

 \begin{theorem}[QM energy asymptotics]
\label{TF theorem 1} \mbox{} \newline
Let  $V$  be a homogenous potential of order $s>2$, define $g=4\pi aN$ and $\omega=g^{-{(s-2)}/{\left(2 s+6\right) }}\Omega$,  and assume that
$N^{-2}g^3\Vert \rho_{g,\Omega }^{\mathrm{TF}}\Vert_{\infty}\to 0$ as $N\to\infty$.

\noindent\hskip .5cm{\rm (i)} If $g \to\infty$  and  $\omega\to 0$ as $N\to\infty$, then $ \stackunder{N\rightarrow {}\infty }{\lim }\left\{ g^{-s/(s+3)}N^{-1} E_{g,\Omega }^{\mathrm{QM}}\left( N\right)  \right\}
=E^{\rm TF}_{1,0}$;

\noindent\hskip .5cm{\rm (ii)} If $g\to\infty$  and $\omega>0$ is fixed as $N\to\infty$, then  $\stackunder{N\rightarrow {}\infty }{\lim }\left\{ g^{-s/(s+3)}N^{-1} E_{g,\Omega }^{\mathrm{QM}}\left( N\right)  \right\}
=\ttfe$;

\noindent\hskip .5cm{\rm (iii)} If $\Omega\to \infty$ and
$\omega\to\infty $ as $N\to\infty$, then $\stackunder{N\rightarrow
{}\infty }{\lim }\left\{\Omega^{- 2s/(s-2)}N^{-1} E_{g,\Omega
}^{\mathrm{QM}}\left( N\right)
 \right\} =\tfeinf$.
 \end{theorem}
Although stated separately,  cases (i) and (ii) can, in fact, be
treated together  because the convergence of the scaled energy is uniform in $\omega$  as long
as
$\omega$ stays bounded.

Besides the energy asymptotics we shall also consider the convergence of the quantum mechanical particle density of  ground states or approximate ground states. We say that a sequence of bosonic, normalized  wave functions $\Psi_N\in L^2(\mathbb R^{3N})$ (depending also on the parameters  $g$ and $\Omega$) is an {\it approximate ground state} if
\beq
\left\langle \Psi _{N},H_{N}\Psi _{N}\right\rangle/ E_{g,\Omega }^{\mathrm{QM}}\left( N\right)\to 1.\eeq(Recall that $E_{g,\Omega }^{\mathrm{QM}}\left( N\right)$ is, by definition, the infimum of the spectrum of \eqref{ham1}.)  The particle density, normalized so that its integral is 1, is defined by
\begin{equation}
\rho _{N,g,\Omega}^{\rm QM}\left( \vec{x}\right) \equiv \int_{\R^{3(N-1)}}\left|
\Psi
_{N}\left( \vec{x},\vec{x}_{2},...,\vec{x}_{N}\right) \right| ^{2}\mathrm{d}%
\vec{x}_{2}....\mathrm{d}\vec{x}_{N}.\label{particle density}
\end{equation}
To ensure convergence of $\rho _{N,g,\Omega}^{\rm QM}\left(
\vec{x}\right) $ we have to rescale it in accord with
\eqref{TFdensityscaling}. For bounded $\omega$ convergence of the density 
follows from the convergence of the energy  using standard arguments
\cite {Griffiths1, LiebSimon1977}.
\begin{theorem}[QM density asymptotics for {\bf $\omega<\infty$}]
    \label{TF theorem 2}
    \mbox{} \\
    Under the conditions of Theorem \ref{TF theorem 1}
(i) or (ii) we have \beq g^{3/(s+3)}\rho _{N,g,\Omega
}^{\mathrm{QM}}\left( g^{1/(s+3)}\vec{x}\right)\to\rho
_{1,\omega }^{\mathrm{TF}}\left(
\vec{x}\right)\label{QMdensityscaling} \eeq in weak $L^1$ sense.
\end{theorem}

The case of ultrarapid rotations, $\omega\to \infty$,  is a little
more delicate. Recall that $\mathcal{M}\subset\mathbb R^3$ was
defined as the set where the  function $V(\xv)-\hbox{$\frac{1}{4}$}r^{2}$
attains its minimum. In Lemma 3.2  it is shown that this set is a
subset of a cylinder and that the scaled TF density
\eqref{TFdensityscalingstrong} is eventually concentrated
on $\mathcal{M}$. In physical terms, the ultrastrong centrifugal
forces  outweigh the repulsion between particles and constrain them
to the boundary of the available region (in the scaled variables).
We show that the same holds for the scaled
QM density.

\begin{theorem}[QM density asymptotics for {\bf $\omega\to \infty$}]
    \label{TF theorem 3}
    \mbox{} \\
    Under the conditions of Theorem \ref{TF theorem 1}
(iii) the scaled particle density \beq\label{qmdens}\Omega^{6/{(s-2)}}\rho
_{N,g,\Omega }^{\mathrm{QM}}\left( \Omega^{
2/{(s-2)}}\vec{x}\right)\eeq becomes
concentrated on the set
$\mathcal{M}$ as $N\rightarrow \infty $ in the sense that the integral of \eqref{qmdens} over any measurable set that has a strictly positive distance  from $\mathcal{M}$ tends to zero.
\end{theorem}

The proofs of Theorems \ref{TF theorem 1}--\ref{TF theorem 3} will be given in several steps. In the next section, we analyze  the asymptotic properties of the
GP and TF theories when $g$ and/or $\Omega$ tend to infinity.  Section 4 contains the upper and lower bounds to the  quantum mechanical energy that complete the proof of Theorem \ref{TF theorem 1}. Theorem \ref{TF theorem 2} is derived from the convergence of the energy by the arguments of \cite {Griffiths1, LiebSimon1977} and Theorem \ref{TF theorem 3} is proved by showing that any mass outside the set ${\mathcal{M}}$ would be in conflict with Theorem \ref{TF theorem 1} (iii).


\section{From GP to TF \label{Section GP to TF}}

In this section we consider the TF limit of GP theory.  Since for $s<\infty$ the GP and TF minimizers spread out over an increasingly large region  as $g$ and/or $\Omega$ tend to infinity we  write
$\xv=\lambda\xvp$ and $\phi(\vec x)=\lambda^{-3/2}\phi'(\xvp)$ with a suitable length scale  $\lambda$ depending on $g$ or $\Omega$ so that the relevant $\vec x'$ remain essentially bounded. We have $\| \phi^{\prime} \|_{2} = \| \phi
\|_{2} = 1 $, and since $V(\vec x)=\lambda^{s}V(\vec x')$ by assumption,  the GP functional can be written
\begin{equation}
    \label{GP functional rescaled}
    \mathcal{E}^{\mathrm{GP}}_{g,\Omega} \left[ \phi \right] = \lambda^{-2} \int_{\RT} \diff \xvp \: \left\{ \left| \left[ \nabla ^{\prime }- i \vec{A}_{\lambda^2\Omega}(\xvp) \right] \phi^{\prime} \right|^{2}+ \lambda^{s+2} V(\xvp) \left| \phi ^{\prime} \right|^{2} - \hbox{$\frac{1}{4}$}(\lambda^2 \Omega)^{2}{r^{\prime}}^{2} \left| \phi ^{\prime} \right|^{2} + g \lambda ^{-1}\left| \phi ^{\prime} \right|^{4} \right\}.
\end{equation}
We now distinguish  two cases. When the rotational contribution to the energy is smaller than or  at most comparable to the interaction energy we equate $\lambda^{s+2}$ with $g\lambda^{-1}$, i.e., choose
\beq\lambda=g^{1/(s+3)}.\label{lambda1}\eeq
For convenience and comparison with \cite{CRY06, CRY07} and \cite{Afta} we define a small parameter $\varepsilon$ by
\beq 1/\varepsilon ^{2} \equiv  \lambda^{s+2}=\lambda^{-1}g=g^{{(s+2)}/{(s+3)}} \label{epsilon_definition}. \eeq
The parameter $\omega=g^{-(s-2)/(2s+6)}\Omega$,  that measures the relative strength of the rotation with respect to the interactions can then be written
\beq\omega=\varepsilon^{{(s-2)}/{(s+2)}}\Omega.
\eeq
We now define a rescaled GP functional $ \tgpf$ by writing
$\mathcal{E}^{\mathrm{GP}}_{g,\Omega} \left[ \phi \right] =\lambda^{-2}\tgpf \left[ \phi' \right]$, or explicitly, dropping the primes,
\begin{equation}
    \label{GP functional rescaled 1}
    \tgpf \left[ \phi \right] \equiv \int_{\RT}\diff\xv \: \left\{ \left| \left[ \nabla -i \magnp \right] \phi \right| ^{2}+\frac{1}{\varepsilon ^{2}}\left( V\left| \phi \right| ^{2}- \frac{\omega^{2}r^{2}}{4}\left| \phi \right| ^{2}+\left| \phi \right|^{4}\right) \right\} .
\end{equation}
In the next subsection we study the $\varepsilon\to 0$ limit of this functional with $\omega<\infty$ fixed or tending to zero.\\

When the rotation dominates the interaction we  take $\lambda^{s+2}$
to be equal to  $(\lambda^2 \Omega)^2$ in \eqref{GP
functional rescaled}, i.e, we take
\beq\lambda=\Omega^{2/{(s-2)}}.\eeq Proceeding as before we now
write $\gpf \left[ \phi \right] =\lambda^{-2}\togpf \left[ \phi'
\right]$ with \beq
    \label{GP functional rescaled 2}
    \togpf \left[ \phi \right] \equiv  \int_{\RT}\diff\xv \: \left\{ \left| \left[\nabla -i \vec A_{\Omega '}
\right] \phi \right| ^{2}+{\Omega'}^2\left[ \left[
V-\hbox{$\frac{1}{4}$}r^{2}\right] \left| \phi \right| ^{2}+\gamma\left| \phi \right|^{4} \right] \right\} \eeq where
\beq{\Omega'}\equiv\Omega^{{(s+2)}/{(s-2)}}\quad
\hbox{and}\quad\gamma=\omega
^{-{2(s+3)}/{(s-2)}}.\eeq The limit
$\Omega\to\infty$, $\omega\to\infty$ (i.e., $\Omega'\to\infty$, $\gamma\to 0$) will be considered in
Subsection 3.2.

\subsection{The Regime $ \omega < {}\infty $}

For $\omega<\infty$ fixed or tending to zero, we study the rescaled GP functional introduced in
\eqref{GP functional rescaled 1} and define $\tgpe $  by
\beq
    \label{GPte}
    \tgpe \equiv  \inf_{\phi \in \gpd, \| \phi \|_2 = 1} \tgpf[\phi]=g^{-{2}/{(s+3)}}E_{g,\Omega
    }^{\mathrm{GP}}.
\eeq The asymptotic behavior of $ g^{-2/(s+3)}\gpe $ as $ g \to {}\infty $ is
then given by the limit of $\tgpe$ as $ \varepsilon
\to 0 $, see (\ref{epsilon_definition}). We state this result
about $ \gpe $ in the following theorem.

\begin{theorem}[GP energy asymptotics for $\omega<\infty$]
    \label{EnAsOmegaconst}
    \mbox{} \\
    Under the conditions of Theorem \ref{TF theorem 1} (i) or
    (ii)  one has, as $ g \to \infty $,
    \beq\label{3.10}
        g^{-s/(s+3)} \gpe = \ttfe + O \left( g^{-{(s+2)}/{(2s+6)}} \log g \right).
    \eeq
\end{theorem}

\begin{proof}
    The proof is obtained by a comparison between suitable lower and upper bounds for $ \tgpe
    $, following closely the proof of Theorem 2.1 in \cite{CRY06}.
    \newline
    The two  regimes of slow (i) and rapid  (ii) rotations can be treated together. Actually, the lower and upper bounds in the case (ii) are sufficient to get the result in the case (i) because the estimates are uniform on bounded intervals of $ \omega$ and the TF ground state energy is a continuous function of $ \omega $.
    To simplify our notation, we denote by $C_\omega$ any
    constant independent of $\varepsilon$. In particular,
    $C_\omega$ needs not be the same from one equation to another, but $ \const $ is always uniformly bounded for $ \omega $ bounded.

    By simply neglecting the positive contribution of the kinetic energy in \eqref{GP functional rescaled 1} one obtains the lower bound
        \beq
        \label{lboundGPtilde}
            \varepsilon^2 \tgpf[\phi] \geq \ttffo[|\phi|^2] \geq \ttfe.
        \eeq
    
        For an upper bound we test the functional  \eqref{GP functional rescaled 1} with a trial function of the form
        \beq
        \label{trial function 1}
            \trial(\xv) = c_{\varepsilon} \sqrt{\modu(\vec x)} \: \cut(\rv) \phase(\rv),
        \eeq
        where $ \phase $ is a phase factor, $ \cut(\rv)$ a  function that vanishes at the singularities of $ \phase $ and $ \modu $ a suitable regularization of $ \ttfm $. Note that both $ \phase $ and  $ \cut $ depend only on the 2d coordinate $ \rv $. For simplicity of notation we have suppressed the dependence on $\omega$ that is regarded as fixed.

        The regularization of $ \ttfm $ is analogous to the one used in \cite{Lieb2}, Lemma 2.3 and is explicitly given by $ \modu \equiv j_\varepsilon\star \ttfm $, with
        \beq
        \label{cut modulus}
            j_\varepsilon(\vec x) \equiv \frac{1}{4\pi \varepsilon^3} \exp \left\{ -\frac{|\vec x|}{\varepsilon} \right\}.
        \eeq
        Since $ \| j_{\varepsilon} \|_1 = 1 $, $ \sqrt{\modu} $ is $L^2-$normalized. It is also clear that $ \modu $ converges uniformly to $ \ttfm $ as $ \varepsilon \to 0 $ and it is uniformly bounded in $ \varepsilon $, i.e., there exists a constant $ \const $ such that $ \modu \leq \const $. Furthermore, although $ \modu $ is not compactly supported, it is exponentially small in $ \varepsilon $ for $ \xv $ sufficiently far from the support of $ \ttfm $. More precisely, denoting\footnote{By \eqref{TF density} the support of $ \ttfm $ is a compact set, uniformly bounded in $\omega$, if  $ \omega $ remains bounded.}
    \beq
    \label{external radius}
        R_{\omega} \equiv \sup \{|\vec x| : \rho^{\rm TF}_{1,\omega}(\vec x)>0\} < \infty,
    \eeq
one has for any $ \xv \in \RT $, $ |\vec x| > R_{\omega} $,
    \beq
    \label{exponential smallness}
        \modu(\xv) = \frac{1}{4\pi \varepsilon^3} \int_{\ttfsupp} \diff \xv^{\prime} \: \exp \left\{ -\frac{\left| \xv - \xv^{\prime} \right|}{\varepsilon} \right\} \ttfm(\xv^{\prime}) \leq \frac{1}{4\pi \varepsilon^3} \exp \left\{ -\frac{|\vec x|-R_{\omega}}{\varepsilon} \right\}.
    \eeq
    We also observe that the gradient of $ \modu $ can be bounded in two ways. By using the fact that $ |\nabla j_{\varepsilon}| = \varepsilon^{-1} j_{\varepsilon} $, one can easily prove that
    \beq
        \label{cut property 1}
        \left| \nabla \modu \right|  \leq \varepsilon^{-1} \left| \modu \right|
    \eeq
    whereas, by exploiting the regularity of $ \ttfm $, i.e., $ \left\| \nabla \ttfm \right\|_1 \leq \const $, one has
    \beq
        \label{cut property 2}
        \int_{\RT} \diff \xv \: \left| \nabla \modu \right| \leq \int_{\RT} \diff \xv \int_{\RT} \diff \xv^{\prime} \: \left| \nabla \ttfm(\xv-\xv^{\prime}) \right| j_{\varepsilon}(\vec x^{\prime}) =  \left\| \nabla \ttfm \right\|_1 \leq \const.
    \eeq
    Writing points $\vec r=(x,y)\in \mathbb R^2$ as complex numbers \( \zeta \equiv x+iy \), the phase $ \phase $ is defined by
        \beq
        \label{phase}
            \phase(\zeta) \equiv \prod_{\zeta_j \in \latt} \frac{\zeta - \zeta_j}{|\zeta- \zeta_j|},
        \eeq
        where \( \latt \) is a square lattice of spacing \( \spac \)
        defined in the following way:
        \beq
        \label{lattice}
            \latt \equiv \left\{ \rv_j = (m \spac, n \spac), \: \: m,n \in \mathbb{Z} \: \Big| \: r < 2R_{\omega}-\spac \right\},
        \eeq
    with $ R_{\omega} $ defined by  \eqref{external radius}.
        We also assume that the spacing is of order
        \( \sqrt{\varepsilon} \), i.e., \( \spac = \delta \sqrt{\varepsilon} \), for some \( \delta > 0 \)
        independent of \( \varepsilon \). Note that the phase \( \phase \) carries
        lines of vortices of degree 1 passing through  the lattice points
        and it coincides with the trial function chosen in
        \cite{CRY06}.
        Moreover the number, $ N_{\varepsilon} $, of lines of vortices included in the support of $ \trial $
        is bounded by $ \const/\varepsilon $, because this support is contained in the finite region $ x \leq R_{\omega} $.     \newline
    The  function $ \cut $ is given by
    \beq
            \cut(\rv) \equiv
            \left\{
            \begin{array}{ll}
                    1   &   \mbox{if} \:\:\:\: |\rv - \rv_j| \geq \varepsilon^{\eta} \:\:\:\: \mbox{for all } \vec r_j\in {\mathcal L}, \\
                    \mbox{} &   \mbox{} \\
                    \displaystyle{\frac{|\rv - \rv_j|}{\varepsilon^{\eta}}}   &   \mbox{if} \:\:\:\: |\rv - \rv_j| \leq \varepsilon^{\eta},
                \end{array}
            \right.
        \eeq
        for some\footnote{This requirement on $ \eta $ is needed in order to apply Theorem 3.1 in \cite{CRY06} (see below in the proof).} \( \eta > 5/2 \).
        \newline
        Finally the constant \( c_{\varepsilon} \) is fixed by the normalization condition, $ \| \trial \|_2 = 1 $, and, since $ \cut \leq 1 $ and $ N_{\varepsilon } \leq C/\varepsilon $,
        \beq
        \label{constant}
            1 \leq c^2_{\varepsilon} \leq  1 + O(\varepsilon^{2\eta-1}) = 1 + o(\varepsilon^4).
        \eeq
        The rescaled functional $ \tgpf$ evaluated on the trial function \eqref{trial function 1} is given by
        \beq
        \label{GPevaluation1}
            \tgpf[\trial] = c_{\varepsilon}^2 \int_{\RT} \diff\xv \: \left| \nabla \left( \cut \sqrt{\modu} \right) \right|^2 + c_{\varepsilon}^2 \int_{\RT} d\xv \: \modu \cut^2 \: \left| \left( \nabla - i \vec{A}_{\omega/\varepsilon} \right) \phase \right|^2 + \frac{\ttffo [ |\trial|^2]}{\varepsilon^2}.
        \eeq
        We  estimate the three energy contributions separately.
    Using that $(a+b)^2\leq 2a^2+2b^2$ the first term can be bounded by\footnote{In this proof $ \mathcal{B}_R $, $ R > 0 $, will always denote a two-dimensional disc of radius $ R $, centered at the origin.}
    \bdm
        c^2_{\varepsilon} \int_{\RT} \diff\xv \: \left| \nabla \left( \cut \sqrt{\modu} \right) \right|^2 \leq 2 c_{\varepsilon}^2 \int_{\RT} \diff\xv \: \left| \nabla \sqrt{\modu} \right|^2 + c_{\varepsilon}^2 \const \int_{\mathcal{B}_{R_{\omega}}} \diff\rv \:  \left| \nabla \cut \right|^2,
    \edm
    where we have used the trivial bound
    \bdm
        \int_{\R} \diff z \: \modu(\xv) \leq \const.
    \edm
    Denoting by $ \bi $ a two-dimensional disc of radius \( \varepsilon^{\eta} \) centered at \( \rv_j \in {\mathcal L}\cap  \mathcal{B}_{R_{\omega}} \), one has
    \beq
        \int_{\mathcal{B}_{R_{\omega}}} \diff\rv \:  \left| \nabla \cut \right|^2 \leq \frac{\left| \bigcup_{j \in \latt} \bi \right| }{\varepsilon^{2\eta}} \leq N_{\varepsilon} \leq \frac{\const}{\varepsilon},
        \eeq
    while, by using both \eqref{cut property 1} and \eqref{cut property 2}, we get
    \beq
        \int_{\RT} \diff\xv \: \left| \nabla \sqrt{\modu} \right|^2 = \int_{\RT} \diff\xv \: \frac{\left| \nabla \modu \right|^2}{4 \modu} \leq \frac{1}{4\varepsilon} \int_{\RT} \diff\xv \: \left| \nabla \modu \right| \leq \frac{\const}{\varepsilon}.
    \eeq
    Hence \eqref{constant} implies the bound
    \beq
    \label{uestimate 1}
        c_{\varepsilon}^2 \int_{\RT} \diff\xv \: \left| \nabla \left( \cut \sqrt{\modu} \right) \right|^2 \leq \frac{\const}{\varepsilon}.
    \eeq
    In order to estimate the second term in \eqref{GPevaluation1}, we first need to restrict the integration to a suitable two-dimensional compact set, by exploiting the exponential smallness of $ \modu $ given by \eqref{exponential smallness}:
    	\bdm
        	\int_{\RT} \diff\xv \: \modu \cut^2 \: \left| \left( \nabla - i \vec{A}_{\omega/\varepsilon} \right) \phase \right|^2 = \int_{|\xv| \leq 2R_{\omega}} \diff\xv \: \modu \cut^2 \: \left| \left( \nabla - i \vec{A}_{\omega/\varepsilon} \right) \phase \right|^2 + 
	\edm
	\bdm
		+ \int_{|\xv| \geq 2R_{\omega}} \diff\xv \: \modu \cut^2 \: \left| \left( \nabla - i \vec{A}_{\omega/\varepsilon} \right) \phase \right|^2 \leq 
	\edm
	\bdm
		\leq \const \int_{\mathcal{B}_{2R_{\omega}}} \diff\rv \: \cut^2 \: \left| \left( \nabla - i \vec{A}_{\omega/\varepsilon} \right) \phase \right|^2 + \frac{\const}{\varepsilon^6} \int_{x \geq 2R_{\omega}} \diff\xv \: \exp \left\{ - \frac{|\vec x| - R_{\omega}}{\varepsilon} \right\} \leq
    \edm
    \bdm
        \leq    \const \int_{\Lambda} \diff\rv \: \left| \left( \nabla - i \vec{A}_{\omega/\varepsilon} \right) \phase \right|^2 + \const \int_{\cup_{i \in \latt} \bi} \diff\rv \: \cut^2 \left| \left( \nabla - i \vec{A}_{\omega/\varepsilon} \right) \phase
        \right|^2 + \frac{\const }{\varepsilon^9} \exp \left\{ - \frac{R_{\omega}}{\varepsilon} \right\},
    \edm
        where 
        \bdm
            \Lambda \equiv \mathcal{B}_{2R_{\omega}} \backslash \bigcup_{j \in \latt} \bi,
        \edm
        and we have used the pointwise estimate (see also (3.10) in \cite{CRY06}),
    	\bdm
        	\left| \left( \nabla - i \vec{A}_{\omega/\varepsilon} \right) \phase \right| \leq \left| \nabla \phase \right| + \frac{\const}{\varepsilon} \leq \frac{\const}{\varepsilon^{3/2}}.
    	\edm
  	which holds for any $ \xv \in \RT $ such that $ |\xv| \geq 2R_{\omega} $.
        
        The second term of the right hand side of the above expression can be bounded in the following way (see also the proof of Theorem 2.1 in \cite{CRY06}):
        \bdm
            \int_{\cup_{i \in \latt} \bi} \diff\rv \: \cut^2 \left| \left( \nabla - i \vec{A}_{\omega/\varepsilon} \right) \phase \right|^2 \leq 2 \int_{\cup_{i \in \latt} \bi} \diff\rv \: \cut^2 \left| \nabla \phase \right|^2 + 2 \int_{\cup_{i \in \latt} \bi} \diff\rv \: \left| \magnp \right|^2 \leq \frac{\const}{\varepsilon} + \frac{\const}{\varepsilon^{4-2\eta}} + \frac{\const}{\varepsilon^{3-4\eta}} \leq \frac{\const}{\varepsilon}.
    \edm
        On the other hand we can apply\footnote{By scaling the expression can be reduced to an integral over a set contained in a ball of radius $ 1 $. The only change with respect to Theorem 3.1 in \cite{CRY06} is then the multiplying factor $ R_{\omega}^2 $.} Theorem 3.1 in \cite{CRY06} to get
        \bdm
            \int_{\Lambda} \diff\rv \: \: \left| \left( \nabla - i \vec{A}_{\omega/\varepsilon} \right) \phase \right|^2 \leq \frac{\pi R_{\omega}^2}{\varepsilon^2}\left(\frac{\omega}{2} - \frac{\pi}{\delta^2} \right)^2 + \frac{\const |\log\varepsilon|}{\varepsilon}
        \edm
        and, by choosing \( \delta = \sqrt{2\pi/\omega} \),
        \beq
        \label{uestimate 2}
            c^2_{\varepsilon} \int_{\RT} \diff\xv \: \modu \cut^2 \: \left| \left( \nabla - i \vec{A}_{\omega/\varepsilon} \right) \phase \right|^2 \leq \frac{\const |\log\varepsilon|}{\varepsilon} + \frac{\const}{\varepsilon} + \frac{\const }{\varepsilon^9} \exp \left\{ - \frac{R_{\omega}}{\varepsilon} \right\} \leq \frac{\const|\log\varepsilon|}{\varepsilon}.
    \eeq
        It remains then to estimate the last term in \eqref{GPevaluation1}, namely the TF energy of $ |\trial|^2 $.  We have
    \bdm
        \ttffo \left[ |\trial|^2 \right] \leq \ttffo \left[ \modu \right] + o(\varepsilon)
    \edm
    and, defining $ \Wo(\xv) \equiv V(\xv) - \omega^2 r^2/4 $,
        \bdm
            \ttffo \left[ \modu \right] - \ttffo \left[ \ttfm \right] = \int_{\RT} \diff \xv \:\: \Wo \left( j_{\varepsilon} \star \ttfm - \ttfm \right) = \int_{\ttfsupp} \diff \xv \:  \ttfm \left( j_{\varepsilon} \star \Wo - \Wo \right).
        \edm
Differentiabilty of $ V $ implies
    \bdm
        \left| \left( j_{\varepsilon} \star \Wo \right)(\xv) - \Wo(\xv) \right| \leq \frac{1}{4\pi} \int_{\RT} \diff \xv^{\prime} \: \left| \Wo(\xv - \varepsilon \xv^{\prime}) - \Wo(\xv) \right| e^{-|\vec x^{\prime}|} \leq \const \varepsilon \int_{\RT} \diff \xv^{\prime} \: r^{\prime} e^{-|\vec x^{\prime}|} \leq \const \varepsilon,
    \edm
    so that
    \beq
        \label{uestimate 3}
            \ttffo \left[ |\trial|^2 \right] \leq \ttffo \left[ \ttfm \right] + \const \varepsilon = \ttfe + \const \varepsilon.
        \eeq
        Putting all the three estimates \eqref{uestimate 1}, \eqref{uestimate 2} and \eqref{uestimate 3} together we obtain the upper bound
    \beq
    \label{ubound}
        \tgpe \leq \frac{\ttfe}{\varepsilon^2} + \frac{\const| \log\varepsilon|}{\varepsilon}.
    \eeq
    The upper bound and the lower bound together give the desired result in the rapid rotation regime (ii).  Indeed, writing $ \varepsilon^2 $ as $ g^{-{(s-2)}/{(s+3)}} $ and using (\ref{GPte}), we get, as $ g \to {}\infty $,
    \bdm
        \ttfe \leq g^{-s/(s+3)} \gpe \leq \ttfe + \const g^{-{(s+2)}/{(2s+6)}} \log g.
    \edm
    The same estimates prove the convergence in the slow rotation regime (i), by uniformity of the bounds and continuity of $ \ttfe $.\end{proof}

As in the two-dimensional case (see \cite{CRY06,CRY07}), a simple
corollary of this result is the $L^1$-convergence of the (scaled) GP
density $ \left| \gpm \right|^2 $ to the TF minimizer $ \ttfm $.
In particular, in the slow rotation regime (i)  this gives the
convergence of the GP density to $ \rho^{{\mathrm{TF}}}_{1,0} $.
In that case the proof of the upper bound could have been
simplified a lot by taking as trial function a suitable
regularization of the (real) TF minimizer without rotation ($
\omega = 0 $).
\newline
Also in the rapid rotation regime (ii) other trial functions could be chosen with the same leading order contribution to the energy as \eqref{trial function 1}. One possibility is to cover the support of $\ttfm$ with Dirichlet boxes and choose a phase factor in each box so that the contribution from $\vec A_{\omega/\varepsilon}$ to the kinetic energy is gauged away to leading order. The reason we have chosen a trial function of the form 
\eqref{trial function 1} is that the error term in \eqref{3.10} has the expected dependence
on the parameters of a next to leading order term, albeit with an unspecified constant in front. In fact, based on 
 considerations of special cases (see, e.g., \cite{Afta})
the true GP minimizer is expected to contain, like \eqref{trial function 1},  a large number $ N_{\varepsilon} \sim \omega/\varepsilon $ of vortex lines, each giving a kinetic contribution of order $ |\log\varepsilon| $ beyond the leading order term. We note that the optimal vortex lattice is expected to be triangular rather than rectangular and bending of lattice lines \cite{Afta1} will also occur but such details only affect the constant factor and higher order contributions in the error term in \eqref{3.10}.

\subsection{The Regime $\omega \to {}\infty $}

We now consider the case of ultrarapid rotations. In the proofs of the quantum mechanical limit theorems in Section 4 we shall not make direct use of the GP minimizers  to construct trial functions for an upper bound to the energy, but the asymptotics  of the TF energy and minimizer as $\omega\to \infty$ will be important. In this subsection we first derive these properties. The  limit theorem
for the GP energy for ultrarapid rotations can be proved by a simple extension of the corresponding TF result and is included for completeness.

We start with a lemma on the structure of the set $\mathcal M$ of minimizing points for the sum of the external potential $V$ and the centrifugal potential $-r^2/4$. As always, $V$ is assumed to be homogeneous of order $s>2$.
\begin{lemma}[${\mathcal M}$ is a subset of a cylinder]\label{form of M}
The set $\mathcal M$ of minimizing points for $W(\vec x)=V(\vec x)-r^2/4$ is a compact
subset of a cylindrical surface with a fixed radial coordinate 
\beq\label{r0} r_0=2\sqrt{\frac{s\,m}{s-2}}\eeq
where $m\equiv-\min W\geq 0$. By scaling it follows that all points of the set $\mathcal M_\Omega$ of minimizing points for $W_\Omega(\vec x)\equiv V(\vec x)-\Omega^2r^2/4$ have the same radial coordinate $r_\Omega=\Omega^{2/(s-2)}r_0$.
\end{lemma}\label{structureofM}
\begin{proof} That $\mathcal M$ is compact follows from continuity of $V$ and the assumption that $W$ tends to $\infty$ as $|\vec x|\to\infty$. On $\mathcal{M}$ we have $\nabla W=0$ and $W=-m$, and thus $\partial
_{r}V=r/2$, $\partial _{z}V=0$ and $V=r^2/4-m$. On the other hand, since $V$ is homogeneous of order $s$, Euler's relation gives $r\partial
_{r}V+z\partial _{z}V=sV$ and hence \eqref{r0}. \end{proof}

If $V$ is rotationally symmetric, i.e., $V(\vec x)=V(r,z)$, and strictly monotonously increasing in $|z|$ (examples: $V(\vec x)=a|\vec x|^s$, or $V(\vec x)=ar^s+b|z|^s$), then $\mathcal M$ is clearly a circle in the $z=0$ plane. It is, however also possible that $\mathcal M$ consists of a two-dimensional subset of the cylinder $r=r_0$ (example: $V(\vec x)=r^s f(|z|/r)$ with $f=1$ on some interval but $>1$ and increasing outside the interval), or of discrete points (example: $V(\vec x)=a|x|^s+b|y|^s+c|z|^s$ with $a\neq b$).

Next we consider the convergence of the TF energy $E^{\rm TF}_{\gamma,1}$ to 
$E^{\rm TF}_{0,1}=\inf W$ 
as $\gamma=
\omega^{-2(s+3)/(s-2)}\to 0$.
\begin{theorem}[TF energy and density for $\omega\to \infty$]\label{lemma asymptotic TF-1}
 \mbox{} \\
For $\gamma \to 0$
\beq\label{gammatozero} E^{\rm TF}_{\gamma,1}=E^{\rm TF}_{0,1}+O(\gamma^{2/5}).\eeq
Moreover, the TF minimizer $\rho^{\rm TF}_{\gamma,1}$ satisfies the bound
\beq  \Vert\rho^{\rm TF}_{\gamma,1}\Vert_\infty\leq 
{\rm const.} \gamma^{-3/5}\eeq
and for any $\epsilon>0$ there is a $\gamma_\epsilon$ such that for $\gamma<\gamma_\epsilon$ 
the support of $\rho^{\rm TF}_{\gamma,1}$ is contained in 
$\mathcal M^{\epsilon}\equiv \{\vec x:\ |\vec x-\vec x'|\leq \epsilon\hbox{ for all }\vec x'\in {\mathcal M}\}$.
\end{theorem}
\begin{proof}
It is clear that $E^{\rm TF}_{\gamma,1}\geq E^{\rm TF}_{0,1}$. For an upper bound we provide a trial function with energy at most $E^{\rm TF}_{0,1}+O(\gamma^{2/5})$.  Let $h$ be any continuous, nonnegative function with support in the unit ball in $\RT$  with $\int h=1$.  For $\delta>0$ and a point $\vec x_0\in \mathcal M$ define $\rho_\delta(\vec x)=\delta^{-3}h((\vec x-\vec x_0)/\delta)$.
Then, using that $W$ is $C^2$ and that $\rho_\delta$ is supported in a ball of radius $\delta$ around $\vec x_0\in{\mathcal M}$, we have
\begin{equation}
E_{\gamma,1}^{\mathrm{TF}}\leq{\mathcal  E}_{\gamma,1}^{\mathrm{TF}}[\rho_\delta]=
\int_{\RT}{\rm
d}\xv\left\{W\rho_\delta+\gamma \rho^2_\delta\right\}\leq E^{\rm TF}_{0,1}+C\delta^2+\gamma\delta^{-3}\Vert h\Vert_2^2.
\label{upperbound_TF_energy0}
\end{equation}
Choosing $\delta=\gamma^{1/5}$ now proves \eqref{gammatozero}.

The TF minimizer is explicitly given by
\begin{equation}
\rho _{\gamma,1}^{\mathrm{TF}}\left( \vec{x}\right) =\dfrac{1}{2%
\gamma}\left[\mu^{\rm TF}_{\gamma,1}-W(\vec x)\right] _{+}.
\label{TF asymptotics 5}
\end{equation}
Since $\rho _{\gamma,1}^{\mathrm{TF}}$ remains normalized as $\gamma\to 0$, it is clear by continuity of $W$ that
$\mu^{\rm TF}_{\gamma,1}$ must converge to the minimum of $W$ and the support of $\rho _{\gamma,1}^{\mathrm{TF}}$ shrinks to the set $\mathcal M$ as stated in the lemma. Moreover, 
$ E_{{%
\gamma},1}^{\mathrm{TF}}\geq E_{0,1}^{\mathrm{TF}}+\gamma||\rho _{\gamma,1}^{\mathrm{TF}}||_{2}^{2}
$
and, by (\ref{TFchempot}),
$
\mu _{g,\Omega }^{\mathrm{TF}%
}= E_{g,\Omega }^{\mathrm{TF}}+g||\rho _{g,\Omega }^{\mathrm{TF}%
}||_{2}^{2} 
$
and this together with   \eqref{gammatozero} implies
\beq\label{334}
0\leq \left[\mu^{\rm TF}_{\gamma,1}-W(\vec x)\right] _{+}\leq \mu^{\rm TF}_{\gamma,1}-E_{0,1}^{\mathrm{TF}}\leq O(\gamma^{2/5})
\eeq
and hence $\Vert\rho^{\rm TF}_{\gamma,1}\Vert_\infty\leq 
{\rm const.} \gamma^{-3/5}$.
\end{proof}
\emph{Remark:\/} In the case that  $\mathcal M$ does not consist of discrete points but is one- or two-dimensional the power of $\gamma$ in the optimal error terms are of higher order than $\gamma^{2/5}$. For instance if $V$ is radially symmetric and
$\mathcal M$ is a circle, the trial function can be taken to be radially symmetric and the error is $O(\gamma^{1/2})$.

As a complementary statement to Theorem \ref{EnAsOmegaconst}
we now prove the convergence of the scaled GP ground state energy
to $\tfeinf $ as $\Omega$ and $\omega\to \infty$. 
\begin{theorem}[GP energy asymptotics for $\omega\to\infty$]
\label{EnAsOmegainfty} \mbox{} \newline As $\Omega\to \infty$ and $\omega\to\infty$
 \beq\label{GP ultrastrong}
\Omega ^{-2s/(s-2)}\gpe =
\tfeinf+{{O}}\left( \Omega ^{\prime -1}+\gamma^{2/5}\right) , \eeq with $\Omega ^{\prime }=
\Omega^{{(s+2)}/{( s-2) }}$ and $\gamma= \omega
^{{-2(s+3)}/{(s-2) }}.$
\end{theorem}

\begin{proof} Note first that $\Omega ^{-2s/(s-2)}\gpe=\Omega^{\prime -2}\togpe$ where $\togpe$ is the ground state energy of the scaled GP functional $\togpf$, c.f.\ \eqref{GP functional rescaled 2}. Dropping the positive kinetic term as in the proof of Theorem
\ref {EnAsOmegaconst} immediately gives the lower bound 
\beq\label{GP_to_TF_Strong} \Omega^{\prime -2}\togpe \geq \tfeinf
. \eeq
For an upper bound we make use of a nonnegative function $h$ with support in the unit ball and $\int h=1$ as in the proof of Theorem
\ref{lemma asymptotic TF-1}, this time requiring $h$ to be  $C^\infty$  so that 
$\Vert \nabla
\sqrt{h}\Vert _{2}<{}\infty $. For $\delta>0$ we put 
$h_\delta(\vec x)=\delta^{-3}h(\vec x/\delta)$.
Let $\vec{x}_{0}\in \mathcal{M}$  and define 
\beq\label{trialfunctionultra}\phi(\vec x)=\sqrt{h_\delta(\vec x-\vec x_0)}\exp\{i\Omega \vec x\cdot(\vec e_z\wedge\vec x_0)/2\}.\eeq
Testing $\togpf$ with this function we obtain
\begin{equation}\Omega ^{\prime -2}\togpe \leq 2\Omega ^{\prime -2}\Vert \nabla \sqrt{%
h_\delta}\Vert _{2}^{2}+\gamma\Vert h_{\delta}\Vert
_{2}^{2}+\int_{\RT}\diff\xv \left( \hbox{$\frac{1}{4}$}\left| \vec{e}_{z}\wedge (\vec{x}_{0}-%
\vec{x})\right| ^{2}+W(\vec x)\right)
h_\delta\left( \vec{x}-\vec{x}_{0}\right) . \label{upper bound
omega infinity}
\end{equation}
Since $h_\delta(\vec{x})\equiv \delta ^{-3}h(\vec{x}/\delta ),$
the first term is ${{O}}\left( \Omega ^{\prime -2}\delta
^{-2}\right) $. Since $\Vert h_\delta\Vert _{2}^{2}\leq \Vert
h_\delta\Vert _{\infty }={{O}}\left( \delta ^{-3}\right)
$, the second term  is ${{O}}\left(
\gamma\delta ^{-3}\right)$. In the last integral we use that $W\in
C^{2}$, that
$\mathrm{supp}\,h_\delta$ is contained in a ball of radius $\delta$ around $\vec x_0$ with $W(\vec x_0)=\tfeinf$, and  $||h_\delta||_{1}=1$ to get 
\beq\label{error estimate}
\int_{\RT}\diff\xv \left( \hbox{$\frac{1}{4}$}\left| \vec{e}_{z}\wedge (\vec{x}_{0}-\vec{x}%
)\right| ^{2}+W(\vec x)\right)
h_\delta\left( \vec{x}-\vec{x}_{0}\right) \leq
\tfeinf+{{O}}(\delta ^{2}).
\eeq
We thus have
\begin{equation}
\Omega ^{\prime -2}\togpe\leq \tfeinf+{{O}}\left(\Omega ^{\prime -2}\delta^{-2}+ \delta
^{2}+\gamma\delta
^{-3}\right) . \label{error estimate 00}
\end{equation}
Equating the second and the last error term leads to the choice 
$\delta=\gamma^{1/5}$ and an error $O({\Omega'}^{-2}\gamma^{-2/5}+\gamma^{2/5})=O(\gamma^{2/5})$ 
provided ${\Omega'}^{-1}\leq \gamma^{2/5}$. For ${\Omega'}^{-1}>\gamma^{2/5}$ we choose $\delta={\Omega'}^{-1/2}$ (this corresponds to equating the first and the second error term in \eqref{error estimate 00}). Then the errors are $O({\Omega'}^{-1})$. Altogether we obtain \eqref{GP ultrastrong}.\end{proof}

\emph {Remark:} The true GP density  is in general not concentrated around a single point in $\mathcal M$  and the trial function \eqref{trialfunctionultra}  is not designed to give optimal error bounds. Another obvious possibility is to replace $h_\delta$ by a regularization of the TF density and choose as phase factor of the \lq giant vortex' type like in \cite{CRY06}. In fact, in the proof of \eqref{QM Upper Boundbis} we use  a function of this form as an ingredient of the many-body trial function. Since Theorem \ref{EnAsOmegainfty} is not directly used for the proof of the corresponding many-body result we do not elaborate on this point further here.

\section{Proofs of the QM Limit Theorems}

In this section we derive the bounds on the quantum mechanical
ground state energy $E_{g,\Omega }^{\mathrm{QM}}$ that lead to the proofs of Theorems \ref{TF theorem 1}--\ref{TF theorem 3}.
The lower bound in the case of ultrarapid rotations is simply
obtained by dropping positive terms from the Hamiltonian.  In the
case $\omega<\infty$ one uses first the diamagnetic inequality
\cite{LL01} to eliminate the vector potential from the Hamiltonian
\eqref{ham2} and then proceeds with the techniques described in
\cite{LSSY05} for the non-rotating case. The upper bound for $\omega<\infty$ is
obtained by  first bounding the QM energy by the GP energy. The
method, that is a generalization of \cite{LiebSeiringerYngvason1},
is described  briefly in \cite{S03} for fixed $g$ and $\Omega$,
but in order to keep track of the error terms as $g$ and/or
$\Omega$ tend to $\infty$ and for completeness we carry it out in
more detail. Once a bound in terms of the GP energy has been
obtained, we can use Theorem \ref{EnAsOmegaconst} of the previous section to relate it to the TF energy. In the regime of ultrarapid rotation, $\omega\to\infty$, 
we use a slightly different method that gives an estimate in terms of the TF energy and error terms involving directly the TF density whose relevant properties were described in Theorem \ref{lemma asymptotic TF-1}. The limit Theorems 2.2-2.3 for the density are simple consequences of the energy bounds and are discussed in Subsection 4.2.

\subsection{Bounds on the QM energy}
\begin{proposition}[Lower bound for  the QM energy]
\label{QM Lower bound} \mbox{} \newline Let the potential $V$ be  homogenous
of order $s>2$. Then \beq\Omega^{%
-2s/(s-2)}N^{-1}E_{g,\Omega }^{\mathrm{QM}}\left( N\right) \geq \tfeinf.\label{lb1}\eeq Furthermore, if $\omega=g^{-(s-2)/(2s+6)}\Omega$ is fixed and $N^{-2}g^3\Vert \rho_{g,\Omega }^{\mathrm{TF}}\Vert_{\infty}\to 0$ as $N\to\infty$ then \beq\stackunder{N\rightarrow {}\infty }{\liminf }\left\{ g^{-\frac{s}{%
s+3}}N^{-1} E_{g,\Omega }^{\mathrm{QM}}\left( N\right)  \right\}
\geq\ttfe\label{lb2}\eeq
uniformly in  $\omega$ on any bounded interval.
\end{proposition}

\begin{proof}
To prove  \eqref{lb1}  consider a normalized  $N$-particle wave function $\Psi _N$ and
let $\hat\rho_N( \vec x)=\lambda^3\rho _{N}\left( \lambda
\vec{x}\right)$ with $\lambda=\Omega^{ 2/{s-2}}$ be the
corresponding scaled density. Then we can write
\begin{equation}
 \Omega ^{-2s/(s-2)}N^{-1}\left\langle
\Psi_{N},H_{N}\Psi_N\right\rangle =C_{\Psi_N}+\stackunder{\RT}{\inf }%
W,  \label{Griffiths proof 6}
\end{equation}
with $W\left( \vec{x}\right) =V\left( \vec{x}\right) -r^{2}/4$
and
\begin{eqnarray}
C_{\Psi _{N}} &\equiv & \Omega  ^{-2s/(s-2)}\left\|
[\nabla -i\vec A_\Omega]\Psi _{N}\right\| _{2}^{2}+\dint_{\RT}\hat{\rho }_{N}(\vec{%
x})\left( W(\vec{x})-\stackunder{\RT}{\inf }W\right)
\mathrm{d}\vec{x}+
\nonumber \\
&&+ \Omega ^{-2s/(s-2)}N^{-1}\stackunder{1\leq i<j\leq N}{\sum
}\left\langle \Psi _{N},v(|\vec{x}_{i}-\vec{x}_{j}|)\Psi
_{N}\right\rangle . \label{Griffiths proof 6bis}
\end{eqnarray}
Since the interaction potential $v$ is by assumption nonnegative the same holds for
$C_{\Psi _{N}}$ so the left hand side of  \eqref{lb1} is $\geq \inf W=\tfeinf$.

Now, let us consider the case when $\omega <\infty$ is fixed as $N\rightarrow {}\infty $.
By the diamagnetic inequality,
$ |(\nabla-i\vec A(\vec x))f(\vec x)|\geq |\nabla|f(\vec x)||$, cf.\ \cite{LL01},
and the bound
\begin{equation}
V\left( \vec{x}
\right)\geq \mu
_{g,\Omega }^{\mathrm{TF}}+\hbox{$\frac{1}{4}$}\Omega ^{2}r^{2}-2g\rho _{g,\Omega }^{\mathrm{TF}}\left( \vec{x}\right)   \label{TF rotating eq 00}
\end{equation}
that follows from Eq.\ \eqref{TF density}
we obtain
\begin{equation}
E_{g,\Omega }^{\mathrm{QM}}\left( N\right) \geq N\mu _{g,\Omega }^{\mathrm{%
TF}}+\stackunder{\Psi, \Vert\Psi\Vert=1 }{\inf }Q\left( \Psi \right) , \label{TF
rotating eq 00bisbis}
\end{equation}
with
\begin{equation}
Q\left( \Psi \right) \equiv \stackrel{N}{\stackunder{i=1}{\sum }}%
\left\| \nabla _{i}\Psi \right\| ^{2}+\stackunder{1\leq i<j\leq N}{\sum }%
\int   v(|\vec{x}_{i}-\vec{x}_{j}|)\left| \Psi \right|
^{2}{\rm d}^{3N}\vec x\, -2g\stackrel{N}{\stackunder{i=1}{\sum
}} \int\rho _{g,\Omega }^{\mathrm{TF}}\left(
\vec{x}_{i}\right) \left| \Psi \right| ^{2} {\rm d}^{3N}\vec x\,  .
\label{TF rotating eq 00bis}
\end{equation}
Since $\rho _{g,\Omega }^{\mathrm{TF}}(\vec x)$ tends to zero for every $\vec x$ as $g\to\infty$, cf. \eqref{TFdensityscalingstrong} it is convenient at this point to carry out a rescaling by writing $\vec x=\lambda \vec x'$ with $\lambda=g^{1/(s+3)}$, cf.\ \eqref{lambda1}.  The scaled interaction potential ${v'}(\vec x')=\lambda ^{2}v(\lambda \vec x')$ has scattering length ${a'}=g^{-1/(s+3)}a$ and the corresponding coupling parameter is
\beq {g'}=4\pi N{a'}=g^{(s+2)/(s+3)}.\eeq
Using \eqref{TFchempot} and \eqref{TF big interaction} we obtain
\begin{equation}
 g^{-s/(s+3)}N^{-1}E_{g,\Omega }^{\mathrm{QM}}-\ttfe\geq \int (\rho^{\rm TF}_{1,\omega})^2 +(N{g'})^{-1}{\inf _{\Psi, \Vert\Psi\Vert=1}}{Q'}( \Psi ) , \label{TF
rotating eq 00bisbis_s}
\end{equation}
where, dropping the primes on the integration variables,
\begin{equation}
{Q'}( \Psi) \equiv \stackrel{N}{\stackunder{i=1}{\sum }}%
\left\| \nabla _{i}\Psi \right\| ^{2}+\stackunder{1\leq i<j\leq N}{\sum }%
\int  {v'}(|\vec{x}_{i}-\vec{x}_{j}|)\left| \Psi \right|
^{2}{\rm d}^{3N}\vec x\, -2{g'}\stackrel{N}{\stackunder{i=1}{\sum
}} \int\rho _{1,\omega }^{\mathrm{TF}}\left(
\vec{x}_{i}\right) \left| \Psi \right| ^{2} {\rm d}^{3N}\vec x\,  .
\label{TF rotating eq 00bis_s}
\end{equation}
We are now exactly in the situation discussed in \cite{LSSY05} for the nonrotating case, cf.\  Eq.\ (6.61) in \cite{LSSY05}.
Like there, the next step is to divide  space into boxes, labeled by $\alpha$ and of side length $l$,  with Neumann boundary conditions and use the lower bound of \cite{LiebYngvason1} for the homogeneous gas in each box.
The result is (cf. Eq.\ (6.62) in \cite{LSSY05})
\beq g^{-\frac{s}{%
s+3}}N^{-1} E_{g,\Omega }^{\mathrm{QM}}\left( N\right)
-\ttfe \geq \int (\rho^{\rm TF}_{1,\omega})^2-\sum_\alpha d_\alpha^2 l^3(1-C{Y'}^{1/17}).\label{411}
\eeq
Here
$d_{\alpha }$ is the maximum value of $\rho^{\rm TF}_{1,\omega}$ in the box $\alpha$ and
${Y'}={a'}^3 N/l^3\sim {g'} N^{-2}/l^3$. Since
\beq gN^{-2}\Vert \rho^{\rm TF}_{g,\Omega}\Vert_\infty={g'} N^{-2}\Vert \rho^{\rm TF}_{1,\omega}\Vert_\infty\eeq
the diluteness condition implies that ${Y'}\to 0$ for fixed $l$. If we  now  first take $N\to\infty$ and then
$\l\to 0$,  the Riemann approximation of $\int (\rho^{\rm TF}_{1,\omega})^2$ implies that the right hand side of \eqref{411} tends to zero, proving \eqref{lb2} for fixed $\omega$. It is also clear that all estimates are uniform in $\omega$ on bounded sets, so in particular one can take $\omega$ to $0$.
\end{proof}


\begin{proposition}[Upper bound on the QM energy for {\bf $\omega<\infty$}]
    \label{QM Upper Bound}
    \mbox{} \\
Let the potential $V$ be homogenous of order $s>2$ and suppose the diluteness condition,
$N^{-2}g^3\Vert \rho_{g,\Omega
}^{\mathrm{TF}}\Vert_{\infty}\to 0$ as $N\to\infty$, is fulfilled.  If $\omega$ is fixed, then
\beq \stackunder{N\rightarrow {}\infty }{\limsup }\left\{ g^{-\frac{s}{%
s+3}}N^{-1} E_{g,\Omega }^{\mathrm{QM}}\left( N\right)  \right\}
\leq\ttfe\eeq uniformly in $\omega$ on bounded intervals.
\end{proposition}

\begin{proof}
    The proof is a combination of a variational bound on the QM energy in terms of the GP energy and the bounds of the GP energy in terms of the TF energy that were  discussed  in Section 3. For the former we can use the same method as in \cite{LiebSeiringerYngvason1} and \cite{Seiringer06-1} and not  all details will be repeated here, but we shall keep track of the error terms and their dependence on the various parameters.

 The main step is to show that under the stated assumptions
    \begin{equation}
        \label{upperQM-GP}
        \frac{E^{\mathrm{QM}}_{g,\Omega}\left(N\right)-NE^{\mathrm{GP}}_{g,\Omega}}{Ng\|\rho _{g,\Omega }^{\mathrm{GP}}\|_\infty} \leq
        o\left(1\right),
    \end{equation}
    by exhibiting a sequence of $N$ particle
    trial functions $\Psi_N$, $N=1,2,\dots$,  such that
    \begin{equation}
    \label{weshow}
        \frac{\langle\Psi_N,H,\Psi_N\rangle\langle\Psi_N,\Psi_N\rangle^{-1}-NE^{\mathrm{GP}}_{g,\Omega}}{Ng\|\rho _{g,\Omega }^{\mathrm{GP}}\|_\infty} \leq o\left(1\right).\;
    \end{equation}
    We write the  trial functions in the form
\beq\label{trial0}
        \Psi_N=F(\vec x_1,\dots,  \vec x_N)G( \vec x_1,\dots,  \vec x_N)
\eeq{with}\beq\label{trial1}G( \vec x_1,\dots,  \vec x_N) \equiv \mbox{$\prod_{i=1}^N$}
        \phi^{\mathrm{GP}}_{g,\Omega}( \vec x_i) \eeq
  and a {\it real} function $F$.    Partial integration, using the variational equation \eqref{GPeq} and the reality of $F$,  leads to
  \begin{multline}
        \label{mini}
        \langle\Psi_N,H,\Psi_N\rangle=N\mu^{\mathrm{GP}}_{g,\Omega}\langle
        \Psi_N,\Psi_N\rangle+\\ \sum_{1\leq i\leq N}\int_{\R^{3N}}{\phantom{}} |\nabla_iF|^2|G|^2
        +\sum_{1\leq i<j\leq
        N}\int_{\R^{3N}}{\phantom{}} v(\xij)|F|^2|G|^2
       - 2g\sum_{1\leq i\leq
        N}\int_{\R^{3N}}{\phantom{}}\rho _{g,\Omega }^{\mathrm{GP}}(\vec x_i)|F|^2|G|^2\;.
    \end{multline}
   The second line of \eqref{mini} is a real quadratic form in $F$  and we shall make use of the fact for an upper bound on the bosonic ground state energy  is not necessary to require that the trial function $F$ is symmetric under permutations of the variables.   This can be seen by a simple adaption of an argument of Lieb \cite{lieb} which implies that the infimum over all functions $F$ is the same as the infimum over all nonnegative, symmetric functions.

   Like in \cite{LiebSeiringerYngvason1} we shall take a trial function of the  Dyson form \cite{dyson}
\begin{equation}
\label{dyson form 1}
 F( \vec x_1,\dots,  \vec x_N)=\mbox{$\prod_{i=1}^N$} F_i( \vec x_1,\dots,  \vec
x_i)
\end{equation}
where
\begin{equation}
  F_i( \vec x_1,\dots, \vec x_i)=f(t_i), \quad
        t_i=\min\left(\xij,j=1,\dots, i-1\right),\label{deft}
\end{equation}
with a function $f$ satisfying
$$
0\leq f\leq 1, \quad
        f'\geq 0.\label{propf}\;
$$
The function $f$ will be  specified shortly. Our estimates  involve the quantities
\begin{equation}\label{IJK}
I\equiv\int_{\RT} (1-f^2),\quad
J\equiv\int_{\RT}
            \left({f'}^2+\hbox{$\frac12$}v\,f^2\right),\quad
            K\equiv \int_{\RT}f f'.
\end{equation}
By exactly the same computation as leads to Eq. (3.29) in \cite{LiebSeiringerYngvason1} we obtain, provided $N\|\rho _{g,\Omega }^{\mathrm{GP}}\|_{\infty}I<1$,
\begin{multline} \label{equation for interaction}\Vert FG\Vert_2^{-2}       \left\{ \int_{\R^{3N}}{\phantom{}} |\nabla F|^2|G|^2
        +\sum_{i<j}\int_{\R^{3N}}{\phantom{}} v(\xij)F^2|G|^2\right\}\leq\\
        \frac 1{{(1-N\|\rho _{g,\Omega }^{\mathrm{GP}}\|_\infty}I)^{2}}\left\{ N^2 J\int_{\RT} {\phantom{}}  \rho _{g,\Omega }^{\mathrm{GP}}(\vec x)^2  +
        \frac{2}{3}N^3{ K^2}\|\rho _{g,\Omega }^{\mathrm{GP}}\|_\infty^2\right\}\end{multline}
and the same technique gives also a bound on the last term in
\eqref{mini}
\begin{equation}
\label{equation for external potential} - 2g\Vert FG\Vert_2^{-2}
\sum_{0<i\leq
        N}\int_{\R^{3N}}{\phantom{}}\rho _{g,\Omega }^{\mathrm{GP}}(\vec x_i)|F|^2|G|^2
        \leq -2gN\int_{\RT} \diff\vec x  \rho _{g,\Omega }^{\mathrm{GP}}(\vec x)^2
        +2g N^2I \|\rho _{g,\Omega }^{\mathrm{GP}}\|_\infty^2.
\end{equation}
   We now choose the function $f$.
 For a parameter $b>a$ that will soon be fixed we define
    \begin{equation}
        \label{deff}
        f(r)=
        \begin{cases}
            (1+\eps_1)u(r)/r &\text{for $r\leq b$}\\ 1
            &\text{for $r>b$}
        \end{cases}
    \end{equation}
    where $u(r)$ is the solution of the scattering equation
    \begin{equation*}
        -u''(r)+\half v(r)u(r)=0 \quad\text{ with $u(0)=0$,
        $\lim_{r\to\infty}u'(r)=1$}
    \end{equation*}
    and $\eps_1$ is determined by requiring $f$ to be continuous.
    Convexity of $u$ gives
    \begin{equation*}
        r\geq u(r)\geq
        \begin{cases}
            0& \text{for $r\leq a$}\\
            r-a& \text{for $r>a$}
        \end{cases}
        ,\quad 1\geq u'(r) \geq
        \begin{cases}
            0& \text{for $r\leq a$}\\
            1-\frac{a}{r}& \text{for $r>a$\;}.
        \end{cases}
    \end{equation*}
    These estimates imply
    \begin{align}
        \label{imply}
        &J\leq (1+\eps_1)^2 4\pi a\\
        &\label{parameter I} I\leq 4\pi\left(\frac{a^3}{3}+ab(b-a)\right)\\
        &K\leq 4\pi
        (1+\eps_1)a\left(b-\frac{a}{2}\right)\\
        &0\leq\eps_1\leq\frac{a}{b-a}.\label{imply2}
    \end{align}

Before proceeding further, we need to relate the supremum of the GP density, $||\rho _{g,\Omega }^{\mathrm{GP}%
}||_{\infty }$, to $||\rho _{g,\Omega }^{\mathrm{TF}%
}||_{\infty }$ since the diluteness condition is stated in terms of the latter. For this purpose we write $\phi^{\mathrm{GP}}_{g,\Omega}=Re^{iS}$ with real $S$ and the nonnegative amplitude $R$. A straightforward computation, using $\nabla\cdot \vec A=0$, gives
    \begin{eqnarray*}
       - (\nabla-i\vec A)^2\phi^{\mathrm{GP}}_{g,\Omega}&=&(-\Delta+2i \vec A\cdot \nabla + A^2
        )\phi^{\mathrm{GP}}_{g,\Omega}\\&=&
        -(\Delta R) e^{iS}-2i(\nabla R)\cdot (\nabla S) e^{iS}+(\nabla
        S)^2Re^{iS}-i(\Delta S) Re^{iS}
        \\&&+2i\vec A\cdot (\nabla R) e^{iS}-2\vec A\cdot (\nabla S) Re^{iS}+ A^2 R e^{iS}
    \end{eqnarray*}
    and from the GP equation (\ref{GPeq}) one obtains
    $$\left(-\Delta +(\nabla S)^2
    -2\vec A\cdot (\nabla S) +A^2 +2g\rho _{g,\Omega }^{\mathrm{GP}}
    +V -\hbox{$\frac{1}{4}$}\Omega^2 r^2\right)R=\mu^{\mathrm{GP}}_{g,\Omega} R.$$
  In any point $\vec x\in\R^3$ where $R$ is maximal $\rho _{g,\Omega }^{\mathrm{GP}}(\vec x)=\|\rho _{g,\Omega }^{\mathrm{GP}}\|_\infty$
    and $\Delta R(\vec x)\leq 0$. Thus,
    $$2g\|\rho _{g,\Omega }^{\mathrm{GP}}\|_\infty\leq-(\nabla S(\vec x))^2
    +2\vec A(x)\cdot \nabla S(\vec x) -A^2(\vec x) -V(\vec x) +\hbox{$\frac{1}{4}$}\Omega^2r^2
    +\mu^{\mathrm{GP}}_{g,\Omega}\;.$$ and since $-(\nabla S)^2
    +2\vec A\cdot \nabla S\leq A^2$ we obtain
    $$2g\|\rho _{g,\Omega }^{\mathrm{GP}}\|_\infty\leq\mu^{\mathrm{GP}}_{g,\Omega}- \inf_{\vec{x}\in\mathbb{R}^3}
    \left\{V(\vec x)-\hbox{$\frac{1}{4}$}\Omega^2 r^2\right\}\;.$$
On the other hand, by  (\ref{TF density}) we have
\begin{equation}
2g||\rho _{g,\Omega }^{\mathrm{TF}}||_{\infty }=\mu _{g,\Omega }^{\mathrm{TF}%
}-\stackunder{\vec{x}\in \Bbb{R}^{3}}{\inf }\left\{ V\left( \vec{x}\right) -%
\hbox{$\frac{1}{4}$}\Omega ^{2}r^{2}\right\}
\end{equation}
and therefore, using \eqref{GPchempot} and \eqref{TFchempot},
\begin{equation}
2g||\rho _{g,\Omega }^{\mathrm{GP}}||_{\infty }\leq \mu _{g,\Omega }^{%
\mathrm{GP}}-\mu _{g,\Omega }^{\mathrm{TF}}+2g||\rho _{g,\Omega }^{\mathrm{TF%
}}||_{\infty }\leq E_{g,\Omega }^{\mathrm{GP}}-E_{g,\Omega }^{\mathrm{TF}%
}+g||\rho _{g,\Omega }^{\mathrm{GP}}||_{\infty }+2g||\rho _{g,\Omega }^{%
\mathrm{TF}}||_{\infty },
\end{equation}
i.e.
\begin{equation}
||\rho _{g,\Omega }^{\mathrm{GP}}||_{\infty }\leq 2||\rho
_{g,\Omega }^{\mathrm{TF}}||_{\infty }\left\{1+ \frac{E_{g,\Omega
}^{\mathrm{GP}}-E_{g,\Omega }^{\mathrm{TF}}}{g||\rho _{g,\Omega }^{\mathrm{TF%
}}||_{\infty }}\right\} . \label{moderate equation 1}
\end{equation}
By using (\ref{TF big interaction}) and Theorem
\ref{EnAsOmegaconst} one sees that
\begin{equation}
\frac{E_{g,\Omega }^{\mathrm{GP}}-E_{g,\Omega
}^{\mathrm{TF}}}{g||\rho
_{g,\Omega }^{\mathrm{TF}}||_{\infty }}=||\rho
_{1,\omega }^{\mathrm{TF}}||_{\infty }^{-1}\left(g^{-s/(s+3)}E_{g,\Omega }^{%
\mathrm{GP}}-E_{1,\omega }^{\mathrm{TF}}\right)= o(1),
\end{equation}
so
\begin{equation}
\label{rho GP leq rho TF}
||\rho _{g,\Omega }^{\mathrm{GP}}||_{\infty }\leq 2||\rho _{g,\Omega }^{%
\mathrm{TF}}||_{\infty }\left\{ 1+o(1)\right\} .
\end{equation}
The diluteness condition $N^{-2}g^{3}||\rho
_{g,\Omega }^{\mathrm{TF}}||_{\infty }\rightarrow 0 $
thus implies the corresponding condition for the GP density,
i.e.,  $N^{-2}g^{3}||\rho _{g,\Omega }^{\mathrm{GP}%
}||_{\infty }\rightarrow 0.$ Therefore,
by choosing
    \begin{equation}
        \label{b}
        b=(N\|\rho _{g,\Omega
        }^{\mathrm{GP}}\|_\infty)^{-\frac{1}{3}},
    \end{equation}
it follows    from (\ref{imply})-(\ref{imply2})  that $\eps_1\rightarrow 0$,
    $\eps_2\equiv N\|\rho _{g,\Omega }^{\mathrm{GP}}\|_\infty I\rightarrow0$ and
    $\eps_3\equiv g^{-1} N^2K^2\|\rho _{g,\Omega }^{\mathrm{GP}}\|_\infty \rightarrow0$
    for $N\rightarrow\infty$.  Altogether one gets from  
    (\ref{mini}), using \eqref{GPchempot}, (\ref{equation for interaction}) and (\ref{equation for external
    potential}),  that
    \begin{eqnarray*}
        N^{-1}E_{g,\Omega }^{\mathrm{QM}}&\leq& E_{g,\Omega }^{\mathrm{GP}}+g\|\rho _{g,\Omega
        }^{\mathrm{GP}}\|^{2}_2
        \left\{{{O}}(\eps_1)+{{O}}(\eps_2)\right\}\;+
        g\|\rho _{g,\Omega }^{\mathrm{GP}}\|_\infty{{O}}(\eps_3)
        \leq E_{g,\Omega }^{\mathrm{GP}}+
       o(1) g\|\rho _{g,\Omega
        }^{\mathrm{GP}}\|_\infty ,
    \end{eqnarray*}
    i.e. (\ref{weshow}).
By (\ref{TFdensityscaling}) we have
\begin{equation}
g^{-s/(s+3)}\left\{ g||\rho _{g,\Omega
}^{\mathrm{TF}}||_{\infty }\right\} =g^{3/(s+3)}||\rho
_{g,\Omega }^{\mathrm{TF}}||_{\infty }=||\rho _{1,\omega
}^{\mathrm{TF}}||_{\infty },
\end{equation}
and by (\ref{rho GP leq rho TF}) and Theorem \ref{EnAsOmegaconst} we can
conclude that
\begin{equation}
g^{-s/(s+3)}N^{-1}E_{g,\Omega }^{\mathrm{QM}}\leq g^{-s/(s+3)%
}E_{g,\Omega }^{\mathrm{GP}}+o(1)||\rho _{1,\omega
}^{\mathrm{TF}}||_{\infty } =E_{1,\omega
}^{\mathrm{TF}}+o\left( 1\right).
\end{equation}\end{proof}
For ultrarapid rotations the proof of \eqref{rho GP leq rho TF} given above is not valid because the error term may blow up as
$\omega\to\infty$.  We shall therefore treat this case separately, using a trial function different from \eqref{trial1}. If a general proof of $||\rho _{g,\Omega }^{\mathrm{GP}}||_{\infty }\leq {(\rm const.)}||\rho _{g,\Omega }^{%
\mathrm{TF}}||_{\infty }$  can be found, then Eq.\ \eqref{weshow}  is verified also for $\omega\to \infty$ and the proof of the next proposition would follow in the same way as the previous one.

\begin{proposition}[Upper bound on the QM energy  for $\omega\to\infty$]
\label{QM Upper Boundbis} \mbox{} \newline Let the potential $V$
be homogenous of order $s>2$ and suppose the diluteness condition,
$N^{-2}g^{3}\Vert \rho _{g,\Omega }^{\mathrm{TF}}\Vert _{\infty
}\to 0$ as $N\to \infty $, is fulfilled. If $\Omega\to\infty$ and $\omega \rightarrow
\infty $ as $N\to\infty$, then
\[
\stackunder{N\rightarrow {}\infty }{\limsup }\left\{\Omega
^{-2s/(s-2)}N^{-1}E_{N,g,\Omega }^{\mathrm{QM}}\left( N\right)
\right\} \leq \tfeinf.
\]
\end{proposition}
\begin{proof}
The first step is to choose a suitable phase factor for the trial function to compensate the vector potential in the kinetic term as far as possible. As shown in Lemma 3.2  the set of
minimizers of
\begin{equation}
W_\Omega\left( \vec{x}\right) =V\left( \vec{x}\right) -\hbox{$\frac{1}{4}$}\Omega
^{2}r^{2} \label{ultra equality 00}
\end{equation}
is a subset of a cylinder with radius $r_{\Omega }>0$. We define the phase factor
as follows:
\begin{equation}
\Theta (\vec{x}_{1},...,\vec{x}_{N})\equiv \prod_{j=1}^{N}\theta (\vec{x}%
_{j}),
\end{equation}
with
\begin{equation}
\theta (\vec{x})\equiv \exp \left\{ i\left[ \frac{1}{2}r_{\Omega
}^{2}\Omega \right] \vartheta \right\} \mathrm{\ for\
}\vec{x}=(r,\vartheta ,z),
\end{equation}
and where $\left[ \cdot\right] $ stands for the integer part.\ Any function $%
\Psi \in L^{2}$ can be written as $\Psi =\Theta \Phi $ with $\Phi \in L^{2}$  and a straightforward computation gives
\begin{eqnarray}
\mathcal{E}_{N,g,\Omega }^{\mathrm{QM}}(\Psi ) &=&\mathcal{E}_{N,g,0}^{\mathrm{QM%
}}(\Phi )+\sum_{j=1}^{N}\int_{\Bbb{R}^{3N}}\left( -\frac{\Omega ^{2}r_{j}^{2}%
}{4}|\Phi |^{2}+\left| \left( -i\nabla _{j}-\vec{A}_{\Omega }(\vec{x}%
_{j})\right) \Theta \right| ^{2}|\Phi |^{2}\right)   \nonumber \\
&&-2\sum_{j=1}^{N}\mathrm{Re}\left( \int_{\Bbb{R}^{3N}}\Phi \Theta
^{*}\left( \left(-i\nabla _{j}-\vec{A}_{\Omega
}(\vec{x}_{j})\right) \Theta \right) \nabla _{j}\Phi ^{*}\right) .
\end{eqnarray}
Since
\begin{equation}
\left( -i\nabla -\vec{A}_{\Omega }\right) \theta =\left(
\frac{1}{r}\left[\frac{1}{2} r_{\Omega }^{2}\Omega \right]
-\frac{r\Omega}{2} \right) \theta \hat e_\vartheta,
\end{equation}
we get, using the Cauchy-Schwartz inequality combined with $2ab\leq
a^{2}+b^{2}$, 
\begin{equation}
\mathcal{E}_{N,g,\Omega }^{\mathrm{QM}}(\Psi )\leq \mathcal{E}_{N,g,0}^{\mathrm{%
QM}}(\Phi )+\sum_{j=1}^{N}\int_{\Bbb{R}^{3N}}\left( -\frac{\Omega
^{2}r_{j}^{2}}{4}|\Phi |^{2}+2\left| \frac{1}{r_{j}}\left[
\frac{1}{2}r_{\Omega }^{2}\Omega \right] -\frac{r_{j}\Omega}{2}
\right| ^{2}|\Phi |^{2}+\left| \nabla _{j}\Phi \right| ^{2}\right) .
\end{equation}
In particular, since $\left[ \frac{1}{2}r_{\Omega }^{2}\Omega
\right] =\frac{1}{2}r_{\Omega }^{2}\Omega +\kappa $ with $|\kappa
|<1$, we then have
\begin{equation}
\mathcal{E}_{N,g,\Omega }^{\mathrm{QM}}(\Psi )-N\inf W_{\Omega}\leq \widetilde{\mathcal{%
E}}_{N,g,\Omega }^{\mathrm{QM}}(\Phi ),  \label{ultra first upper
bound}
\end{equation}
with
\begin{equation}\label{tildeeqm}
\widetilde{\mathcal{E}}_{N,g,\Omega }^{\mathrm{QM}}(\Phi )\equiv
\sum_{j=1}^{N}\int_{\Bbb{R}^{3N}}\left( 2\left| \nabla _{j}\Phi
\right| ^{2}+\left( \widetilde{W}_{\Omega}-\inf W_{\Omega}\right) |\Phi |^{2}\right)
+\sum_{1\leq i<j\leq
N}\int_{\Bbb{R}^{3N}}v(|\vec{x}_{i}-\vec{x}_{j}|)|\Phi |^{2}
\end{equation}
and
\begin{equation}
\widetilde{W}_{\Omega}\left( \vec{x}\right) \equiv V\left( \vec{x}\right)
+\Omega
^{2}\left( -\frac{r^{2}}{4}+\left| \frac{r_{\Omega }^{2}}{r}-r\right| ^{2}+%
\frac{4}{\Omega ^{2}r^{2}}\right) .
\end{equation}
The functional  \eqref{tildeeqm} describes the QM energy of a non-rotating system of
particles with mass $1/2$ in the positive external potential
$(\widetilde{W}_{\Omega}-\inf W_{\Omega})$ and with a two-body interaction potential $v\geq 0$.

We need to chose a trial function for \eqref{tildeeqm}.  As as in the proof of the previous proposition we are dealing with a {\it real} quadratic form so the infimum over all $\Phi\in L^2({\mathbb R}^3)$ is the same as the infimum over all symmetric $\Phi\in L^2({\mathbb R}^3)$. Thus, by \eqref{ultra first upper bound},
\begin{equation}
E_{N,g,\Omega }^{\mathrm{QM}}-N\inf W_{\Omega}\leq \stackunder{\Phi }{\inf }\frac{%
\widetilde{\mathcal{E}}_{N,g,\Omega }^{\mathrm{QM}}(\Phi )}{||\Phi
||_{2}^{2}}. \label{ultra first upper boundbis}
\end{equation}
We denote by
$\rho _{\epsilon }\equiv j_{\epsilon }\star \rho _{g,\Omega
}^{\mathrm{TF}}$ the regularized TF density with $j_{\epsilon
}$ defined by (\ref{cut modulus}) for any $\epsilon >0$.\footnote{This $\epsilon$ is unrelated to the $\varepsilon$ defined in \eqref{epsilon_definition}.}
Observe that $||\rho _{\epsilon }||_{1}=1$ and $||\nabla
\sqrt{\rho _{\epsilon }}||_{2}=\epsilon ^{-2}$
since $|\nabla j_{\epsilon }|=\epsilon ^{-1}j_{\epsilon }$ and $%
||j_{\epsilon }||_{1}=1$ for any $\epsilon >0$. Our
trial function is defined as
\begin{equation}
{\tilde \Phi }(\vec{x}_{1},...,\vec{x}_{N})\equiv F(\vec{x}_{1},...,\vec{x%
}_{N})G(\vec{x}_{1},...,\vec{x}_{N}),  \label{second test
function}
\end{equation}
with
\begin{equation}
G(\vec{x}_{1},...,\vec{x}_{N})\equiv \stackunder{j=1}{\stackrel{N}{\prod }}%
\sqrt{\rho _{\epsilon }(\vec{x}_{j})}  \label{ultra definition
G}
\end{equation}
while the function $F$ is of the Dyson form, cf. (\ref{dyson form
1})-(\ref {deft}) and (\ref{deff}). 

The estimation of $\widetilde{\mathcal{E}}_{N,g,\Omega }^{\mathrm{QM}}(\tilde \Phi )/{||\tilde \Phi ||_{2}^{2}}$ follows closely the computations in \cite{LiebSeiringerYngvason1}, Eqs.\ (3.11)-(3.29), but with the regularized TF density
$\rho _{\epsilon }$ instead of the GP density. The diluteness condition 
$N^{-2}g^3||\rho _{g,\Omega
}^{\mathrm{TF}}||_{\infty }\to 0$ implies that the same condition is also fulfilled for $\rho _{\epsilon }$ because
\begin{equation}
||\rho _{\epsilon }||_{\infty }\leq
||\rho _{g,\Omega }^{\mathrm{TF}}||_{\infty
}||j_{\epsilon }||_{1}=||\rho _{g,\Omega
}^{\mathrm{TF}}||_{\infty }. \label{ultra diluteness condition}
\end{equation}
The GP equation, that was used in the computations in  \cite{LiebSeiringerYngvason1} to obtain Eq.\ (3.28) in that paper, is not at our disposal for $\rho _{\epsilon }$ but we can instead use the Cauchy-Schwarz inequality combined with $2ab\leq a^{2}+b^{2}$. In this way,  using also the positivity of $v$, we obtain
\begin{eqnarray}
E_{N,g,\Omega }^{\mathrm{QM}}-N\inf W_{\Omega} &\leq &\sum_{j=1}^{N}\left\|
FG\right\| _{2}^{-2}\left\{ 4\int_{\Bbb{R}^{3N}}\left| \nabla
_{j}G\right| ^{2}F^{2}+\int_{\Bbb{R}^{3N}}\left(
\widetilde{W}_{\Omega}\left( \vec{x}_{j}\right)
-\inf W_{\Omega}\right) F^{2}G^{2}\right\}   \nonumber \\
&&+4\sum_{j=1}^{N}\left\| FG\right\| _{2}^{-2}\left\{ \int_{\Bbb{R}%
^{3N}}\left| \nabla _{j}F\right| ^{2}G^{2}+\sum_{i<j}\int_{\Bbb{R}^{3N}}v(|%
\vec{x}_{i}-\vec{x}_{j}|)F^{2}G^{2}\right\} . \label{ultra first
upper boundbisbis}
\end{eqnarray}
Since $|\nabla j_{\epsilon }|=\epsilon ^{-1}j_{\epsilon }$  we have for the first term of (\ref{ultra first
upper boundbisbis}) the bound
\begin{equation}
4\left\| FG\right\| _{2}^{-2}\int_{\Bbb{R}^{3N}}\left| \nabla
_{j}G\right| ^{2}F^{2}=4\left\| FG\right\|
_{2}^{-2}\int_{\Bbb{R}^{3N}}\frac{1}{\rho _{\epsilon }\left(
\vec{x}_{j}\right) }\left| \frac{\nabla _{j}\rho _{\epsilon
}\left( \vec{x}_{j}\right) }{2\sqrt{\rho _{\epsilon }\left(
\vec{x}_{j}\right) }}\right| ^{2}F^{2}G^{2}\leq
\frac{1}{\epsilon ^{2}}.
\end{equation}
Using Eqs.\ (3.15)-(3.16) and (3.21) in \cite{LiebSeiringerYngvason1} as well as $(\widetilde{W}_{\Omega}-\inf W_{\Omega})\geq 0$, we obtain the bound
\begin{equation}
\left\| FG\right\| _{2}^{-2}\int_{\Bbb{R}^{3N}}\left(
\widetilde{W}_{\Omega}\left( \vec{x}_{j}\right) -\inf W_{\Omega}\right)
F^{2}G^{2}\leq \frac{1}{1-N||\rho _{\epsilon }||_{\infty
}I}\int_{\Bbb{R}^{3}}\mathrm{d}^{3}\vec{x}\left(
\widetilde{W}_{\Omega}-\inf W_{\Omega}\right) \rho _{\epsilon }
\end{equation}
with $I$ defined by (\ref{IJK}), provided $N||\rho _{\epsilon }||_{\infty
}I<1$ that is guaranteed by the diluteness condition.
 The last two terms in
(\ref{ultra first upper boundbisbis}) are bounded in exactly the same
way which leads to (\ref{equation for interaction}) with $\rho
_{\epsilon } $ in the place of $\rho _{g,\Omega }^{\mathrm{GP}}$.  We omit the details.   Altogether
we have the upper bound
\begin{equation}
N^{-1}E_{N,g,\Omega }^{\mathrm{QM}}-\inf W_{\Omega}\leq \frac{1}{\epsilon ^{2}}%
+\left( 1+o\left( 1\right) \right) \int_{\Bbb{R}^{3}}\mathrm{d}\vec{x}%
\left\{ \left( \widetilde{W}_{\Omega}-\inf W_{\Omega}\right) \rho _{\epsilon
}+4g\rho _{\epsilon }^{2}\right\} +o\left( 1\right) g||\rho _{\epsilon
}||_{\infty }  .  \label{ultra diluteness
condition 2}
\end{equation}
Since $\rho_\epsilon$ is normalized, 
\begin{equation}
||\rho _{\epsilon }||_{2}\leq ||\rho _{\epsilon }||_{\infty}\leq  ||\rho _{g,\Omega }^{\mathrm{%
TF}}||_{\infty}
\end{equation}
by \eqref{ultra diluteness condition}. Also, 
\begin{equation}
 \Omega^{-2s/(s-2)}\inf W_{\Omega}=E_{0,1}^{%
\mathrm{TF}},  \label{ultra equality 0}
\end{equation}
and by Lemma \ref{form of M} all minimizing points have the same radial coordinate $r_{\Omega }=r_{0} \Omega
 ^{\frac{2}{s-2}}$ with $r_{0}$ the radius of the set $%
\mathcal{M}$. The inequality (\ref{ultra diluteness
condition 2}) now implies
\begin{eqnarray}
\Omega  ^{-2s/(s-2)}N^{-1}E_{N,g,\Omega }^{\mathrm{QM%
}}-E_{0,1}^{\mathrm{TF}} &\leq &\Omega ^{-2s/(s-2)}\left\{ \int_{\Bbb{R}^{3}}%
\mathrm{d}\vec{x}\left( \widetilde{W}_{\Omega}-\inf W_{\Omega}\right) \rho
_{\epsilon
}\right\} \left( 1+o\left( 1\right) \right)   \nonumber \\ &&+5\Omega  ^{-2s/(s-2)}g||\rho _{g,\Omega }^{%
\mathrm{TF}}||_{\infty}\left( 1+o\left( 1\right) \right)   \nonumber \\
&&+\Omega  ^{-2s/(s-2)}\epsilon ^{-2}. \label{ultra
equality 1}
\end{eqnarray}
We now have to
bound each term of the right hand side. For the second term we use that
\begin{equation}
\Omega ^{-2s/(s-2)}g||\rho _{g,\Omega
}^{\mathrm{TF}}||_{\infty }=\gamma||\rho _{\gamma,1
}^{\mathrm{TF}}||_{\infty }\leq o\left( 1\right)  \label{ultra
equality 2}
\end{equation}
by (\ref{TFdensityscalingstrong}) and because $||\rho _{\gamma,1
}^{\mathrm{TF}}||_{\infty }\leq O(\gamma^{-3/5})$ by Lemma 3.2. The last term in \eqref{ultra equality 1} is $o(1)$ as long as $\epsilon\gg \Omega^{-s/(s-2)}$.

It remains to consider the first term in \eqref{ultra
equality 1}. By scaling we have
\begin{eqnarray}\label{4.58}
\Omega  ^{-2s/(s-2)}\int_{\Bbb{R}^{3}}\mathrm{d}%
\vec{x}\left( \widetilde{W}_{\Omega}-\inf W_{\Omega}\right) \rho _{\epsilon } &=&\int_{\Bbb{%
R}^{3}}\mathrm{d}\vec{x}\left\{ V\left( \vec{x}\right) -\hbox{$\frac{1}{4}$}%
r^{2}+\left| \frac{r_{0}^{2}}{r}-r\right| ^{2}+\frac{4}{\Omega
^{\prime
2}r^{2}}\right\} \widetilde{\rho }_{\epsilon }\left( \vec{x}\right) -E_{0,%
1}^{\mathrm{TF}},  \nonumber \\
&&
\end{eqnarray}
with
\begin{equation}
\widetilde{\rho }_{\epsilon }(\vec x)=\Omega ^{\frac{6}{s-2}%
}\rho _{\epsilon }\left(\Omega ^{\frac{2}{s-2}}%
\vec{x}\right) =j_{\widetilde{\epsilon }}\star \rho _{\gamma,1%
}^{\mathrm{TF}}(\vec x),\qquad \widetilde{\epsilon }\equiv \Omega  ^{-\frac{2}{s-2}%
}\epsilon . 
\end{equation}
As $\gamma\to 0$ the support of $\rho _{\gamma,1%
}^{\mathrm{TF}}$ becomes concentrated on the set $\mathcal M$ where $V\left( \vec{x}\right) -\hbox{$\frac{1}{4}$}%
r^{2}$ is minimized and all points in $\mathcal M$ have the same radial coordinate, $r_0$. Moreover, outside of the support of $\rho _{\gamma,1%
}^{\mathrm{TF}}$ the regularized density $\widetilde{\rho }_{\epsilon }$ decreases exponentially if $d(\vec x)/{\widetilde \epsilon}\to\infty$ where $d(\vec x)$ is the distance of $\vec x$ from the support of $\rho _{\gamma,1%
}^{\mathrm{TF}}$ (see also the proof of Theorem 3.1). This implies that the right hand side of \eqref{4.58} tends to zero as $\gamma\to 0$, $\Omega\to \infty$ and $\epsilon\to 0$ with $\epsilon^{-1}\Omega^{-s/(s-2)}\to 0$.
\end{proof}


\subsection{Convergence of the QM particle density\label{Section
Griffiths}}
\subsubsection{The case $\omega<\infty$}
\emph{Proof of Theorem \ref{TF theorem 2}:} We use Griffiths' argument
\cite {Griffiths1} in the same way as for an analogous problem in
\cite{LiebSimon1977}.
Take any bounded function $f:\RT\to \R$ and for any $%
\sigma \in \left[ \delta ,\delta \right] $ ($\delta >0$) perturb
the Hamiltonian $H_{N}$ with the external potential $\sigma
g^{s/(s+3)}f( g^{-1/(s+3)} \vec{x}).$ Because $f$ is
bounded, the statements (i)-(ii) of Theorem \ref{TF theorem 1} can
also be proven for the perturbed external potential $ \{V\left(
\vec{x}\right) +\sigma g^{s/(s+3)}f( g^{-1/(s+3)}
\vec{x})
\}$. Namely, if $g\to\infty$  and $\omega\geq0$ is fixed then the corresponding ground state $%
E_{g,\Omega ,\sigma }^{\mathrm{QM}}\left( N\right) $ converges to
\begin{equation}
\stackunder{N\rightarrow +\infty }{\lim }\left\{
g^{-s/(s+3)}N^{-1} E_{g,\Omega ,\sigma }^{\mathrm{QM}}\left(
N\right) \right\} =E_{1,\omega ,\sigma }^{\mathrm{TF}}\mathrm{\
for\ any\ }\sigma \in \left[ \delta ,\delta \right] ,\mathrm{\
}\delta
>0 \label{Griffiths proof 1}
\end{equation}
where $E_{1,\omega ,\sigma }^{\mathrm{TF}}$ is the TF energy with $V$ replaced by $V+\sigma f$.
Consequently, for any approximated ground states $\Psi _{N}$ of
$H_{N}$ we have
\begin{equation}
\left\langle \Psi _{N},\left[ H_{N,\sigma }-H_{N}\right] \Psi
_{N}\right\rangle =Ng^{s/(s+3)}\sigma \left\langle f,\widetilde{\rho }%
_{N}\right\rangle ,  \label{Griffiths proof 2}
\end{equation}
with $\widetilde{\rho }_{N}( \vec{x}) =g^{3/(s+3)}\rho _{N}(
g^{1/(s+3)}\vec{x}) $. Hence, the Rayleigh-Ritz principle
for any $\sigma \in \left[ \delta ,\delta \right] $ leads to
\begin{equation}
E_{g,\Omega ,\sigma }^{\mathrm{QM}}\left( N\right) -E_{g,\Omega ,0}^{\mathrm{%
QM}}\left( N\right) \leq Ng^{s/(s+3)}\sigma \left\langle
f,\widetilde{\rho
}_{N}\right\rangle +\left\langle \Psi _{N},\left[ H_{N}-E_{g,\Omega }^{%
\mathrm{QM}}\left( N\right) \right] \Psi _{N}\right\rangle .
\label{Griffiths proof 3}
\end{equation}
Because (\ref{TF big interaction}) and (\ref{TFdensityscaling})
are still valid with $V( \vec{x}) $ replaced by $V( \vec{x})
+\sigma g^{s/(s+3)}f(g^{-1/(s+3)}\vec{x}) ,$ the
previous inequality combined with (\ref{Griffiths proof 1})
implies in the limit $N\rightarrow \infty $ that
\begin{equation}
\sigma ^{-1}\left[ E_{1,\omega ,\sigma }^{\mathrm{TF}}-E_{1,\omega
,\sigma
}^{\mathrm{TF}}\right] \leq \stackunder{N\rightarrow +\infty }{\lim \inf }%
\left\langle f,\widetilde{\rho }_{N}\right\rangle \mathrm{\ for\;any\ }%
\sigma \in \left( 0,\delta \right] ,  \label{Griffiths proof 4}
\end{equation}
whereas a negative parameter $\sigma \in \left[ -\delta ,0\right)
$ reverses the inequality. The proof of the differentiability of
$E_{1,\omega ,\sigma
}^{\mathrm{TF}}$ at $\sigma =0$ is deduced from similar estimations as from (%
\ref{Griffiths proof 2}) to (\ref{Griffiths proof 4}) combined
with the continuity of the function $\sigma \mapsto \left\langle
f,\rho _{1,\omega
,\sigma }^{\mathrm{TF}}\right\rangle ,$ where $\rho _{1,\omega ,\sigma }^{%
\mathrm{TF}}$ is the minimizer of the variational problem
$E_{1,\omega
,\sigma }^{\mathrm{TF}}$. We omit the details. In other words, we have $%
\partial _{\sigma }E_{1,\omega ,\sigma }^{\mathrm{TF}}=\left\langle f,\rho
_{1,\omega }^{\mathrm{TF}}\right\rangle $. Consequently, by
(\ref{Griffiths
proof 4}) and its reversed inequality, we obtain Theorem \ref{TF theorem 2}%
.{}\\ {\phantom a} {\hfill $\square$}
\subsubsection{The case $\omega\to\infty$}

In the case of ultrarapid rotations Griffiths' argument
 is not as easily applicable as in the previous situation. There are two complications. First, perturbing $V$ with a scaled additional term $\sigma f$ leads to a potential that is not homogeneous and the proof of the upper bound for the QM energy has to be modified in order to get  (\ref{Griffiths proof 4}) with $E_{0,1,\sigma
}^{\mathrm{TF}}$ replacing $E_{1,\omega ,\sigma }^{\mathrm{TF}}$. Secondly, 
 if $\mathcal{M}$ consists of more than one point the variational problem for $E_{0,1,\sigma }^{\mathrm{TF}}$
does not have a unique minimizer and  $E_{0,1%
,\sigma }^{\mathrm{TF}}$ is not differentiable at $\sigma =0$ in general. To get around this complication, we would need to
perturb the Hamiltonian with an additional term to 
avoid the degeneracy of the variational problem. An example of such an argument is given in \cite{LS06}. But since we are content with proving that the density is concentrated on $\mathcal M$ in the limit and not striving to obtain the exact limiting measure on $\mathcal M$ we can ignore all these problems and use the following simpler argument.

\emph{Proof of Theorem \ref{TF theorem 3}:\/} If $\mathcal K$ is any set with a positive distance from $\cal M$, then $W(\vec x)-E^{\rm TF}_{0,1}\geq c>0$ on ${\mathcal K}$ with some strictly positive number $c$. Hence, by Eqs.\ (\ref{Griffiths proof 6})-(\ref{Griffiths proof 6bis}),
\beq \label{K}\Omega ^{-2s/(s-2)}N^{-1}\left\langle
\Psi_{N},H_{N}\Psi_N\right\rangle-E^{\rm TF}_{0,1}\geq c\int_{\mathcal K}\hat \rho_N.
\eeq
On the other hand, by Proposition 4.3 we know that the left hand side of \eqref{K} tends to zero as $N\to\infty$, so $\int_{\mathcal K}\hat \rho_N\to 0$.{}\\ {\phantom a} {\hfill $\square$}
\subsection{Remarks on \lq flat' trapping potentials}
As mentioned in the Section 2 the case of a \lq flat' trap, that corresponds formally to $s=\infty$, can be treated in essentially the same way as we have done for $s<\infty$. By a flat potential we mean that $V$ is 0 inside some open, bounded set $\mathcal B$ with a smooth boundary and $\infty$ outside. More precisely, the kinetic term in the many-body Hamiltonian \eqref{ham2} and the GP functional \eqref{GP functional} are defined with Dirichlet conditions on the boundary of $\mathcal B$, but Neumann conditions lead, in fact, to the same results in the large $g$, large $\Omega$ limit.

For $s=\infty$  Eq.\ \eqref{omega} resp.\ \eqref{gamma} reduces to $\omega=g^{-1/2}\Omega$ resp.\ $\gamma=\Omega^{-2}g$ and the TF energy scales as $g^{-1}E^{\rm TF}_{g,\Omega}=E^{\rm TF}_{1,\omega}$, resp.\ $\Omega^{-2}
E^{\rm TF}_{g,\Omega}=E^{\rm TF}_{\gamma,1}$. Theorem 2.1 becomes in case (i)
$ \left\{ g^{-1}N^{-1} E_{g,\Omega }^{\mathrm{QM}}\left( N\right)  \right\}
\to E^{\rm TF}_{1,0}$, in case (ii) $ \left\{ g^{-1}N^{-1} E_{g,\Omega }^{\mathrm{QM}}\left( N\right)  \right\}
\to E^{\rm TF}_{1,\omega}$, and in case (iii) $ \left\{ \Omega^{-2}N^{-1} E_{g,\Omega }^{\mathrm{QM}}\left( N\right)  \right\}
\to E^{\rm TF}_{0,1}$.

The proofs require only some minor modifications. For instance, in the proofs of Therorem 3.1 and Propositions 4.2--4.3 the trial functions have to be modified in order to take the boundary condition into account.  In the case of ultrarapid rotations it has to be noted that the set $\mathcal M$ is now a subset of the boundary $\partial \mathcal B$ of $\mathcal B$,  consisting of the points on $\partial \mathcal B$ where the centrifugal potential $-r^2/4$ is minimal, i.e., where $r$ is maximal. If the boundary is smooth one can still use a Taylor expansion for the proof of an analogue of Theorem 3.3 (with different error terms).

\appendix

\let\oldsect\section
\def\section#1{\def\thesection{Appendix \Alph{section}} \oldsect{#1}\def\thesection{\Alph{section}}}

\section{The TF density at large rotational velocities\label{Section Appendix}}
The TF density is explicitly given by \eqref{TF density}, i.e., 
\beq\rho _{g,\Omega }^{\mathrm{TF}}\left( \vec{x}\right)
=\dfrac{1}{2g}\left[
\mu _{g,\Omega }^{\mathrm{TF}}+\hbox{$\frac{1}{4}$}\Omega ^{2}r^{2}-V\left( \vec{x}%
\right) \right] _{+}.\eeq
To get a picture how it changes with the parameters, in particular as 
$\omega=g^{-(s-2)/(2s+6)}\to\infty$, it is convenient to use the scaling
\beq \Omega^{6/{(s-2)}}\rho _{g,\Omega
}^{\mathrm{TF}}\left(\Omega^{
2/{(s-2)}}\vec{x}\right)=\rho _{\gamma,1
}^{\mathrm{TF}}\left(\vec{x}\right) \eeq
to eliminate the dependence of the potential on $\Omega$ and consider
\beq\label{rhogamma1}\rho _{\gamma,1}^{\mathrm{TF}}\left( \vec{x}\right)
=\dfrac{1}{2\gamma}\left[
\mu _{\gamma,1 }^{\mathrm{TF}}+\hbox{$\frac{1}{4}$}r^{2}-V\left( \vec{x}%
\right) \right] _{+}\eeq with $\gamma=\omega^{-2(s+3)/(s-2)}$. 
As $\omega$ increases from 0 to $\infty$, $\gamma$ decreases from $+\infty$ to zero and, due to the normalization,   the chemical potential  $\mu _{\gamma,1 }^{\mathrm{TF}}$ decreases monotonically from $+\infty$ to $E^{\rm TF}_{0,1}=\min\left\{V\left( \vec{x}%
\right)-\hbox{$\frac{1}{4}$}r^{2}\right\}$. Since $V$ is homogeneous of order $s>2$ it is clear that $E^{\rm TF}_{0,1}<0$. By continuity there is a $\gamma_c$ (and a corresponding $\omega_c$) such that $\mu _{\gamma_c,1 }^{\mathrm{TF}}=0$. Explicitly,
\beq
\gamma_c=\hbox{$\frac 12$}\int\left[
\hbox{$\frac{1}{4}$}r^{2}-V\left( \vec{x}%
\right) \right] _{+}{\rm d}\vec x.
\eeq
Since $V( 0)=0$ and $V$ is continuous, we have $\rho _{\gamma,1}^{\mathrm{TF}}( \vec{x})>0$ in a neighborhood of $ 0$ for $\gamma>\gamma_c$ (i.e., $\omega<\omega_c$).  For $\gamma< \gamma_c$, on the other hand,
$\rho _{\gamma,1}^{\mathrm{TF}}( \vec{x})=0$ in a cylinder around the $z$ axis, because  $\mu _{\gamma,1 }<0$ and $V\geq 0$. In other words, for  $\gamma<\gamma_c$  (i.e., $\omega>\omega_c$), the centrifugal force creates
a \lq hole' of radius $\geq 2\sqrt{-\mu _{\gamma,1 }}$ in the density.

To describe this in a little more detail we  introduce cylindrical coordinates $(r,z,\vartheta)$ and consider the density as a function of $r$ at fixed $(z,\vartheta)$, first for $z=0$. Using homogeneity of $V$ we see that the boundary of the support of $\rho _{\gamma,1}^{\mathrm{TF}}$ is determined by the solutions to the equation
\beq\label{A5}a({\vartheta})r^s-r^2/4= \mu _{\gamma,1 }^{\mathrm{TF}}\eeq
with $a(\vartheta)=V(1,0,\vartheta)>0$. If $\gamma>\gamma_c$ (i.e., $\omega<\omega _c$) so $\mu _{\gamma,1 }^{\mathrm{TF}}>0$, the equation has one solution, $r(\vartheta)^+_\gamma>0$, and $\rho _{\gamma,1}^{\mathrm{TF}}(r,0,\vartheta)>0$ for $r<r(\vartheta)^+_\gamma$ but $\rho _{\gamma,1}^{\mathrm{TF}}=0$ outside.
Eq.\ \eqref{A5} has two solutions, $r(\vartheta)^{\pm}_\gamma$, if $\gamma_c> \gamma>\gamma_\vartheta$ (i.e., $\omega_c< \omega<\omega_\vartheta$) where $\gamma_\vartheta\geq 0$ is determined by solving \eqref{A5} together with
\beq sa(\vartheta) r^{s-1}-r^2/2=0.\eeq
The density $\rho _{\gamma,1}^{\mathrm{TF}}(r,0,\vartheta)$ is nonzero for $r$ inside the interval
$(r(\vartheta)^{-}_\gamma,r(\vartheta)^{+}_\gamma)$ but vanishes outside. For $\gamma=\gamma_\vartheta$ the interval shrinks to a point with radial coordinate
\beq\label{rlim}
r(\vartheta)^{\rm lim}=(2sa(\vartheta))^{-1/(s-2)}=\sup_{\gamma\geq \gamma_\vartheta} r(\vartheta)^{-}_\gamma=\inf_{\gamma\geq \gamma_\vartheta} r(\vartheta)^{+}_\gamma.
\eeq
The chemical potential $\mu _{\gamma_\vartheta,1 }^{\mathrm{TF}}$ and hence $\gamma_\vartheta$
is determined by inserting \eqref{rlim} into \eqref{A5}. If $\gamma_\vartheta>0$, then $\rho _{\gamma,1}^{\mathrm{TF}}$  is identically zero in the $\vartheta$ direction for $\gamma_\vartheta\geq \gamma\geq 0$.  If $\gamma_\vartheta=0$, on the other hand, then \eqref{A5} implies that $(r_\vartheta^{\rm lim},0,\vartheta)\in{\mathcal M}$ and thus $r(\vartheta)^{\rm lim}=r_0$ by Lemma \ref{form of M}.

For $z\neq 0$ we can use the homogeneity of $V$ to write $V(\vec x)=r^sV(1,z/r,\vartheta)$a and apply the considerations above to a ray with fixed $z/r$ and $\vartheta$. It is, however, more interesting to consider what happens at fixed $z\neq 0$ in the case that $V$ is monotonically increasing in $|z|$, for instance if $V(\vec x)=V_0(\vec r)+c|z|^s$.   At fixed $z$ the term $c|z|^s$ acts as a shift of the chemical potential, $\mu _{\gamma,1 }^{\mathrm{TF}}\to \mu _{\gamma,1 }^{\mathrm{TF}}-c|z|^s$.  Hence for a \lq hole' to appear at  $z> 0$ it is not necessary that $\mu _{\gamma,1 }^{\mathrm{TF}}$ becomes negative, it appears already when
$\mu _{\gamma,1 }^{\mathrm{TF}}-c|z|^s<0$. For any $\gamma<\infty$ the density has therefore a \lq hole' for sufficiently large $|z|$, the width increasing with $|z|$. As  $\gamma$ decreases from an initial value $>\gamma_c$  the \lq hole'  moves down to lower values of $z$, reaching $z=0$ at $\gamma=\gamma_c$.

As an example, consider $V(r,z,\vartheta)=r^s(1-\varepsilon \sin^2 \vartheta)+|z|^s$ with
$0<\varepsilon<1$. Here $\gamma_\vartheta>0$ for all $\vartheta\in [0,2\pi)$ except for
$\vartheta=0$ and $\vartheta=\pi$. As $\gamma\to 0$ the density converges to 
$\left\{\half \delta(x-1)+\half\delta(x+1)\right\}\delta(y)\delta(z)$. If  $\varepsilon=0$, i.e., $V=r^s+c|z^s|$, then $\gamma_\vartheta=0$ for all $\vartheta$ and the limiting density for $\gamma\to 0$ is $(2\pi)^{-1}\delta(r-1)\delta(z)$.

\let\section\oldsect
\vspace{.5cm}
\noindent
{\bf Acknowledgments.} 
This work was supported by an Austrian Science Fund (FWF) grant P17176-N02. JY would like to thank the Niels Bohr International Academy and Nordita, Copenhagen, for hospitality and Chris Pethick, Gentaro Watanabe  and Gordon Baym for discussions.

\end{document}